\documentclass[a4paper,11pt]{article}
\usepackage[utf8]{inputenc}
\usepackage{amssymb}
\usepackage{amsfonts,amsmath,amsthm}
\usepackage{mathtools,braket}
\usepackage[margin=0.95in]{geometry}
\usepackage{enumerate,enumitem}
\usepackage{multirow, multicol}
\usepackage{thmtools, thm-restate}
\usepackage[plain]{algorithm}
\usepackage{algpseudocode,color}
\usepackage[multiple]{footmisc}

\usepackage[backend=biber,style=alphabetic,maxalphanames=4,maxcitenames=2,maxbibnames=99,natbib=true]{biblatex}
\usepackage{varwidth}
\usepackage{url}
\usepackage{soul} %for ul command  
\usepackage[breaklinks,hyperfootnotes=false]{hyperref}
\hypersetup{
	colorlinks   = true, %Colours links instead of ugly boxes
	urlcolor     = blue, %Colour for external hyperlinks
	linkcolor    = blue, %Colour of internal links
	citecolor   = blue %Colour of citations
}

\newtheorem{theorem}{Theorem}
\newtheorem{lemma}[theorem]{Lemma}
\newtheorem{proposition}[theorem]{Proposition}

\newtheorem{claim}{Claim}
\newtheorem{definition}{Definition}

\newcommand{\maxRegret}{\text{maxRegret}}
\newcommand{\regret}{\text{regret}}
\newcommand{\OMV}{\text{OMV}}
\newcommand{\abs}[1]{\left|#1\right|}
\newcommand{\specialcell}[2][c]{%
  \begin{tabular}[#1]{@{}c@{}} #2\end{tabular}}

\DeclareMathOperator*{\argmin}{arg\,min}

\algnewcommand\algprocedure{\textbf{Procedure:}}
\algnewcommand\Procedurename{\item[\underline{\algprocedure}]}
\algnewcommand\algmain{\textbf{Main:}}
\algnewcommand\Main{\item[\underline{\algmain}]}
\algnewcommand\algorithmicinput{\textbf{Input:}}
\algnewcommand\Input{\item[\algorithmicinput]}
\algnewcommand\algorithmicoutput{\textbf{Output:}}
\algnewcommand\Output{\item[\algorithmicoutput]}

\urldef{\email}\path|{vijay.menon, kate.larson}@uwaterloo.ca|

\allowdisplaybreaks

 \title{Mechanism Design for Locating a Facility under Partial Information}
\author{
	Vijay Menon\footnote{David R.\ Cheriton School of Computer Science, University of Waterloo.\ \email} \\
	\newline
	\and
	Kate Larson\footnotemark[1]
}

\date{}

\begin{document}

 \maketitle

\begin{abstract}
 We study the classic mechanism design problem of locating a public facility on a real line. In contrast to previous work, we assume that the agents are unable to fully 
specify where their preferred location lies, and instead only provide coarse information---namely, that their preferred location lies in some interval. Given such partial 
preference information, we explore the design of \textit{robust} deterministic mechanisms, where by \textit{robust} mechanisms we mean ones that perform well with respect to 
all the possible unknown true preferred locations of the agents. Towards this end, we consider two well-studied objective functions and look at implementing these under two 
natural solution concepts for our setting \emph{i)} very weak dominance and \emph{ii)} minimax dominance.  We show that under the former solution concept, there are no 
mechanisms that do better than a naive mechanism which always, irrespective of the information provided by the agents, outputs the same location. However, when using the 
latter, weaker, solution concept, we show that one can do significantly better, and we provide upper and lower bounds on the performance of mechanisms for the objective 
functions of interest. Furthermore, we note that our mechanisms can be viewed as extensions to the classical optimal mechanisms in that they perform optimally when agents 
precisely know and specify their preferred locations.
\end{abstract}

\section{Introduction}

% We consider the classic problem of locating a public facility on a real line or an interval, a canonical problem in \textit{mechanism design without money}. In the 
% standard version 
% of the facility location problem, there are $n$ agents, denoted by the set $[n] = \{1, \cdots, n\}$, and each agent $i \in [n]$ has an ideal location $x_i$ on the real line 
% (or some bounded subinterval of the same), which is their most preferred location for the public facility. The cost of an agent for a facility located at $p$ is given by 
% $C(x_i, p) = |p - x_i|$, the distance from the facility to the agent's ideal location, and the task in general is to locate a facility that minimizes some objective function. 
% The most commonly considered objective functions are \emph{a)} sum of costs for the agents and \emph{b)} the maximum cost for an agent. In the mechanism design 
% version of the problem, the main question is to see if the objective under consideration can be implemented, either optimally or approximately, in (weakly) dominant 
% strategies.

We consider the classic problem of locating a public facility on a real line or an interval, a canonical problem in \textit{mechanism design without money}. In the 
standard version 
of this problem, there are $n$ agents, denoted by the set $[n] = \{1, \cdots, n\}$, and each agent $i \in [n]$ has a preferred location $x_i$ for the public facility. The cost 
of an agent for a facility located at $p$ is given by 
$C(x_i, p) = |p -x_i|$, the distance from the facility to the agent's ideal location, and the task in general is to locate a facility that minimizes some objective function. 
The most commonly considered objective functions are \emph{a)} sum of costs for the agents and \emph{b)} the maximum cost for an agent. In the mechanism design 
version of the problem, the main question is to see if the objective under consideration can be implemented, either optimally or approximately, in (weakly) dominant 
strategies.

While the standard version of the problem has received much attention, with several different variants like extensions to multiple facilities (e.g., \citep{procaccia13,lu10}), 
looking at alternative objective functions (e.g., \citep{feldman13,cai16}) etc.\ being extensively studied, the common assumption in this literature is that the agents are 
always precisely aware of their preferred locations on the real line (or the concerned metric space, depending on which variant is being considered). However, this might not 
always be the case and it is possible that the agents currently do not have accurate information on their ideal locations, or their preferences in general. To illustrate this, 
imagine a simple scenario where a city wants to build a school on a particular street (which we assume for simplicity is just a line) and aims to build one at a location that 
minimizes the maximum distance any of its residents have to travel to reach the school.. While each of the residents is able to specify which block they would like the school 
to be located at, some of them are unable to precisely pinpoint where on the block they would like it because, for example, they do not currently have access to information 
(like infrastructure data) to better inform themselves, or they are simply unwilling to put in the cognitive effort to refine their preferences further. Therefore, instead of 
giving a specific location $x$, they end up giving an interval $[a, b]$, intending to say ``\textit{I know that I prefer the school to be built between the points $a$ and $b$, 
but I am not exactly sure where I want it.}''  

The above described scenario is precisely the one we are concerned about in this paper. That is, in contrast to the standard setting of the facility location problem, we 
consider the setting in which the agents are uncertain (or partially informed) about their own true locations $x_i$ and the only information they have is that their preferred 
location $x_i \in [a_i, b_i]$, where $b_i - a_i \leq \delta$ for some parameter $\delta$ which models the amount of inaccuracy in the agents' reports. Now, given such 
partially informed agents, our task is to look at the problem from the perspective of a designer whose goal is to design ``robust'' mechanisms under this 
setting. Here by ``robust'' we mean that, for a given performance measure and when considering implementation under an appropriate solution concept, the mechanism should 
provide good guarantees with respect to this measure for all the possible underlying unknown true locations of the agents. The performance measure we use here is 
based on the minimax regret solution criterion, which, informally, for a given objective function, $S$, is an outcome that has the ``best worst case'', or one that induces 
the least amount of regret after one realizes the true input.\footnote{We refer the reader to appendix~\ref{app:whymR} for a discussion on the choice of regret as the 
performance measure.} More formally, if $\mathcal{P} = [0, B]$ denotes the set of all points where a facility can be located and $\mathcal{I} = [a_1, b_1] \times \cdots 
\times [a_n, b_n]$ denotes the set of all the possible vectors that correspond to the true ideal locations of the agents, then the minimax optimal solution, $p_{opt}$, for 
some objective function $S$ (like the sum of costs or the maximum cost) is given by 
\begin{equation*} 
%   p_{opt} = \argmin_{p \in \mathcal{P}}\max_{I \in \mathcal{I}} \left(S(I, p) - \min_{p' \in \mathcal{P}} S(I, p')\right),
p_{opt} = \argmin_{p \in \mathcal{P}}\underbrace{\max_{I \in \mathcal{I}} \left(S(I, p) - \min_{p' \in \mathcal{P}} S(I, p')\right)}_{\text{$\maxRegret(p, \mathcal{I})$ }},
 \end{equation*}
where $S(I, p)$ denotes the value of $S$ when evaluated with respect to $I \in \mathcal{I}$ and a point $p$. 

Thus, our aim is to design mechanisms that approximately implement the optimal minimax value (i.e., $\maxRegret(p_{opt}, \mathcal{I})$) w.r.t.\ two objective 
functions---average cost and maximum cost---and under two solution concepts---very weak dominance and minimax dominance---that naturally extend to our setting (see 
Section~\ref{sec:prelims} for definitions). In particular, we focus on deterministic and anonymous mechanisms that additively approximate the optimal minimax 
value, and our results are summarized in Table~\ref{tab1}.

Before we move on to the rest of the paper, we anticipate that a reader might have some questions, especially w.r.t.\ our choice of performance measure and our decision to 
use 
additive as opposed to multiplicative approximations. We try to preemptively address these briefly in the section below.

\begin{table*}[!t]
\begin{center}
\resizebox{0.9\textwidth}{!}{
\begin{tabular}{c|c|c|c|c|}
% \footnotesize
\cline{2-5}
& \multicolumn{2}{|c|}{Average cost}    &\multicolumn{2}{|c|}{Maximum cost} \\ \cline{2-5} 
& Upper bound         & Lower bound    & Upper bound     & Lower bound   \\ \hline
\multicolumn{1}{|c|}{\multirow{4}{*}{\specialcell{very weak\\ dominance}}} & \multirow{4}{*}{$\frac{B}{2}$} & \multirow{4}{*}{\specialcell{$\frac{B}{2}$ \\ 
{[}Thm.~\ref{thm:a-vwd}{]}}} & \multirow{4}{*}{$\frac{B}{2}$} & \multirow{4}{*}{\specialcell{$\frac{B}{2}$ \\ {[}Thm.~\ref{thm:a-vwd-mc}{]} }} \\ 

\multicolumn{1}{|c|}{} & & & & \multicolumn{1}{|c|}{} \\ 
\multicolumn{1}{|c|}{} & & & & \multicolumn{1}{|c|}{} \\
\multicolumn{1}{|c|}{} & & & & \multicolumn{1}{|c|}{} \\
\multicolumn{1}{|c|}{\multirow{4}{*}{\specialcell{minimax\\ dominance}}}   & \multirow{4}{*}{\specialcell{$\frac{3\delta}{4}$ \\ {[}Thm.~\ref{thm:deltmedian}{]} }} &  
\multirow{4}{*}{\specialcell{$\frac{\delta}{2}$ \\ (only for mechanisms \\ with finite range) \\ {[}Thm.~\ref{thm:lb-ac}{]} }}&  
\multirow{4}{*}{\specialcell{$\frac{B}{4} + 
\frac{3\delta}{8}$\\  {[}Thm.~\ref{thm:phantomHalf}{]} }} &  \multirow{4}{*}{\specialcell{$\frac{B}{4}$\\ \citep[Thm.~5]{golomb17}}} \\ 
\multicolumn{1}{|c|}{} & & & & \multicolumn{1}{|c|}{} \\ 
\multicolumn{1}{|c|}{} & & & & \multicolumn{1}{|c|}{} \\
\multicolumn{1}{|c|}{} & & & & \multicolumn{1}{|c|}{} \\\hline
\end{tabular}}
\end{center}
 \caption{Summary of our results. All the bounds are with respect to deterministic mechanisms.}
 \label{tab1}
\end{table*}

\subsection{Some Q \& A.}
\textbf{{Why regret?}} We argue below why this is a good measure by considering some alternatives.
 \begin{enumerate}
  \item Why not bound the ratio of the objective values of \emph{a)} the outcome that is returned by the mechanism and \emph{b)} the optimal outcome for that input? This, for 
instance, is the approach taken by \citet{chiesa12}. In our case this is not a good measure because we can quickly see that this ratio is always unbounded in the worst-case. 

\item Why not find a bound $X$ such that for all $I \in \mathcal{I}, S(I,p) - S(I,p_I ) \leq X$, where $p$ is the outcome of the mechanism and $p_I$ is the optimal solution 
associated with $I$? This, for instance, is the approach  taken by \citet{chiesa14}. Technically, this is essentially what we are doing when using max.\ regret. However, 
using 
regret is more informative because if we make a statement of the form $\maxRegret(p, I) - \maxRegret(p_{opt}, I) \leq Y$, then this conveys two things: \emph{a)} for any $p'$ 
there is at 
least one $I \in \mathcal{I}$ such that $S(I, p') - S(I, p_I ) \geq Z$, where $Z = \maxRegret(p_{opt})$ (i.e. it gives us a sense on what is achievable at all---which in turn 
can be thought of as a natural lower bound) and \emph{b)} the point $p$ chosen by the mechanism is at most $(Y + Z)$-far from the optimal objective value for any $I$. Hence, 
to 
convey these, we employ the notion of regret. We refer the reader to Appendix~\ref{app:whymR} for a slightly more elaborate discussion.
 \end{enumerate}
 
 \noindent\textbf{Why additive approximations?} We use additive as opposed to multiplicative approximations because one can see that when using the latter and  w.r.t.\ the 
max.\ cost objective function both the solution concepts that we consider in this paper---which we believe are natural ones to consider in this setting---do not provide any 
insight into the problem as there are no bounded mechanisms. Again, we refer the reader to Appendix~\ref{app:addvsmul} for a more elaborate discussion.

\subsection{Related work}

There are two broad lines of research that are related to the topic of this paper. The first is, naturally, the extensive literature that focuses on designing mechanisms in 
the context of the facility location problem and the second is the work done in mechanism design which considers settings where the agents do not completely specify their 
preferences. Below, beginning with the latter, we describe some of the papers that are most relevant to our work.

\textbf{Designing mechanisms with incomplete preferences.} A disproportionate amount of the work in mechanism design considers settings where the agents have 
complete information about their preferences. However, as one might expect, the issue of agents not specifying their complete preferences has been considered in the mechanism 
design literature and the papers that are most relevant to this paper are the series of papers by \citet*{chiesa12,chiesa14,chiesa15}, and the 
works of \citet{hyafil07a,hyafil07}. Below we briefly discuss about each of them.

The series of papers by \citet{chiesa12,chiesa14,chiesa15} considers settings where the agents are uncertain about their own types and they look at this model in the context 
of auctions. In particular, in their setting the only information agents have about their valuations is that 
it is contained in a set $K$, where $K$ is any subset of the set of all possible valuations.\footnote{\citet{chiesa14} argue that their model is equivalent to the Knightian 
uncertainty model that has received much attention in decision theory (see related works section in \cite{chiesa14} and the references therein). However, here we do not use 
the term Knightian uncertainty, but instead just say that the agents are partially informed. This is because, the notion we use here, which we believe is the natural one to 
consider in the context of our problem, is less general than the notion of Knightian uncertainty.} 
Under this setting, \citet{chiesa12} look at single-item auctions 
and they provide several results on the fraction of maximum social welfare that can be achieved under implementation in very weakly dominant and undominated strategies; 
subsequently, \citet{chiesa14} study the performance of VCG mechanisms in the context of combinatorial auctions when the agents are uncertain about their own 
types and under undominated and regret-minimizing strategies; and finally, \citet{chiesa15} analyze the Vickrey mechanism in the context of multi-unit auctions and, again, 
when the agents are uncertain about their types, and in this case they essentially show that it achieves near-optimal performance (in terms of social welfare) under 
implementation in undominated strategies. The partial information model that we use in this paper is inspired by this series of papers. In particular, our 
prior-free and absolute worst-case approach under partial information is similar to the one taken by \citet{chiesa12,chiesa14,chiesa15} (although such absolute worst-case 
approaches are not uncommon and have been previously considered in many different settings). However, our work is also different from theirs in that, unlike 
auctions, the problem we consider falls within the domain of \textit{mechanism design without money} and so their results do not carry over to our setting. 

% \footnotetext{\citet{chiesa14} argue that their model is equivalent to the Knightian uncertainty model that has received much attention in decision theory (see related works 
% section in \cite{chiesa14} and the references therein). However, here we do not use the term Knightian uncertainty, but instead just say that the agents are partially 
% informed. This is because, the notion we use here, which we believe is the natural one to consider in the context of our problem, is less general than the notion of Knightian 
% uncertainty.}

The other set of papers that are most relevant to the broad theme here is the work of \citet{hyafil07a,hyafil07} who considered the problem of 
designing mechanisms that have to make decisions using partial type information. Their focus is again on contexts where payments are allowed and in \citep{hyafil07a} 
they mainly show that a class of mechanisms based on the minimax regret solution criterion achieves approximate efficiency under approximate dominant strategy 
implementation. In \citep{hyafil07} their focus is on automated mechanism design within the same framework. While the overall theme in both their works is similar to ours, 
i.e., to look at issues that arise when mechanisms only have access to partial information, the questions they are concerned with and the model used are different. For 
instance, in the context of the models used, whereas in ours and Chiesa et al.'s models the agents do not know their true types and are 
therefore providing partial inputs, to the best of our understanding, the assumption in the works of \citet{hyafil07a,hyafil07} is that the mechanism has access to partial 
types, but agents are aware of their true type. This subtle change in turn leads to the focus being on solution concepts that are different from ours. 

In addition to the papers mentioned above, note that another way to model uncertain agents is to assume that each of them has a probability distribution 
which 
tells them the probability of a point being their ideal location. For instance, this is the model that is used by \citet{feige11} in the context of task scheduling. However, 
in our model the agents do not have any more information than that they are within some interval, which we emphasize is not 
equivalent to assuming that, for a given agent, every point in the its interval is equally likely to be its true ideal location. 

% Finally, we also note the paper by who look designing single-player mechanisms under Knightian uncertainty model in the context of the rent-extraction problem (see . 

\textbf{Related work on the facility location problem.} Starting with the work of \citet{moulin80} there has been a flurry of research looking at designing strategyproof 
mechanisms (i.e., mechanisms where it is a (weakly) dominant strategy for an agent to reveal her true preferences) for the facility location problem. These can be 
broadly 
divided into two branches. The first one consists of work, e.g., \citep{moulin80,barbera94,schummer02,masso11,dokow12}, that focuses on 
characterizing the class of strategyproof mechanisms in different settings (see \citep{barbera01} and \citep[Chapter 10]{nisan07} for a survey on some of these results). The 
second 
branch consists of more recent papers which fall under the broad umbrella of \textit{approximate mechanism design without money}, initially advocated by \citet{procaccia13}, 
that focus on looking at how well a strategyproof mechanism can perform under different objective functions \citep{procaccia13,lu10,feldman13,fotakis16,feigenbaum16}. Our 
paper, which tries to understand the performance of mechanisms under different solution concepts and objective functions when the agents are partially informed about 
their own locations, falls under this branch of the literature.

\section{Preliminaries} \label{sec:prelims}

Recall that in the standard (mechanism design) version of the facility location problem there are $n$ agents, denoted by the set $[n] = \{1, \cdots, n\}$, and each agent $i 
\in [n]$ has a true preferred\footnote{We often omit the term ``preferred'' and instead just say that $\ell^*_i$ is agent $i$'s location.} location $\ell^*_i \in [0, B]$, for 
some 
fixed\footnote{Note that here we make the assumption that the domain under consideration is bounded instead of assuming that the agents can be anywhere on the real line. This 
is necessary only because we are focusing on additive approximations instead of the usual multiplicative approximations. (For a slightly more elaborate explanation, see the 
introduction section of the paper by \citet{golomb17}.)}~constant $B \in \mathbb{R}$. A vector $I = (\ell_1, \cdots, \ell_n)$, where $\ell_i \in [0, B]$, is referred to as a 
location profile  and the cost of agent $i$ for a facility located at $p$ is given by $C(\ell^*_i, p) = |p - \ell^*_i|$ (or equivalently, their utility  is $-|p - 
\ell^*_i|$), the distance from the facility to the agent's location.\footnote{This particular utility function that is considered here is equivalent to the notion of 
symmetric single-peaked preferences that is often used in the economics literature (see, e.g., \cite{masso11}).} In general, the task in the facility location problem is to 
design mechanisms---which are, informally, functions 
that map location profiles to a point (or a distribution over points) in $[0, B]$---that (approximately) implement the outcome associated with a particular objective 
function. 

In the version of the problem that we are considering, each agent $i$, although they have a true location $\ell_i^* \in [0, B]$, is currently unaware of their true location 
and 
instead only knows an interval $[a_i, b_i] \subseteq [0, B]$ such that $\ell_i^* \in [a_i, b_i]$. The interval $[a_i, b_i]$, which we denote by $K_i$, is referred 
to as the \textit{candidate locations} of agent $i$, and we use $\mathbb{K}_i$ to denote the set of all possible candidate locations of agent $i$ (succinctly referred to as 
the set of candidate locations). Now, given a profile of the set of candidate locations $(\mathbb{K}_1, \cdots, \mathbb{K}_n)$, we have the following definition.

\begin{definition}[$\delta$-uncertain-facility-location-game]
 For all $n \geq 1$, $B > 0$, and $\delta \in [0, B]$, a profile of the set of candidate locations $(\mathbb{K}_1, \cdots, \mathbb{K}_n)$ is said to induce a  
$\delta$-uncertain-facility-location-game if, for each $i$, $\mathbb{K}_i = \{ [a_i, b_i] \mid b_i - a_i \leq \delta$ and $[a_i, b_i] \subseteq [0, 
B]\}$ (or in words, for 
each agent $i$, their set of candidate locations can only have intervals of length at most $\delta$).
\end{definition}

% \textbf{Remark:} We refer to $\delta$ as the inaccuracy parameter and in general do not assume, unless otherwise stated, that the designer knows this $\delta$. 
% Additionally, note that in the definition above if $\delta = 0$, then we have the standard facility location setting where the set of candidate locations associated with 
% every agent 
% is just a set of points in $[0, B]$. For a given profile of candidate locations $(K_1, 
% \cdots, K_n)$, we say that ``the reports are exact'' when, for each agent $i$, $K_i$ is a single point and not an interval. 

\textbf{Remark:} We refer to $\delta$ as the inaccuracy parameter. In general, when proving lower bounds we assume that the designer knows this $\delta$ as this only 
makes our results stronger, whereas for positive results we explicitly state what the designer knows about $\delta$. Additionally, note that in the definition above if 
$\delta = 0$, then we have the standard facility location setting where the set of candidate locations associated with every agent is just a set of points in $[0, B]$. For a 
given profile of candidate locations $(K_1, \cdots, K_n)$, we say that ``the reports are exact'' when, for each agent $i$, $K_i$ is a single point and not an interval. 

\subsection{Mechanisms, solution concepts, and implementation} \label{sec:solutionConcepts}
A (deterministic) mechanism $\mathcal{M} = (X, F)$ in our setting consists of an action space $X = (X_1, \cdots, X_n)$, where $X_i$ is the action space associated with agent 
$i$, and an outcome function $F$ which maps a profile of actions to an outcome in $[0, B]$ (i.e., $F\colon X_1\times\cdots\times X_n \to [0, B]$). A mechanism is 
said to be \textit{direct} if, for all $i$, $X_i = \mathbb{K}_i$, where $\mathbb{K}_i$ is the set of all possible candidate locations of agent $i$. For every 
agent $i$, a strategy is a function $s_i\colon \mathbb{K}_i \to X_i$, and we use $\Sigma_i$ and $\Delta(\Sigma_i)$ to respectively denote the set of all pure and mixed 
strategies of agent $i$.  

Since the outcome of a mechanism needs to be achieved in equilibrium, it remains to be defined what equilibrium solution concepts we consider in this paper. Below we define, 
in the order of their relative strengths, the two solution concepts that we use here. We note that the first (very weak dominance) was also used by \citet{chiesa12}. 

\begin{definition}[{very weak dominance}]
 In a mechanism $\mathcal{M} = (X, F)$, an agent $i$ with candidate locations $K_i$ has a very weakly dominant strategy  
$s_i \in \Sigma_i$ if $\forall s'_{i} \in \Sigma_i, \forall \ell_i \in K_i$, and $\forall s_{-i} \in \Sigma_{-i}$,
 \begin{equation*}
  C\left(\ell_i, F(s_i(K_i), s_{-i}(K_{-i}))\right) \leq C\left(\ell_i, 
F(s'_i(K_i), s_{-i}(K_{-i}))\right).
 \end{equation*}
\end{definition}
In words, the above definition implies that for agent $i$ with candidate locations $K_i$, it is always best for $i$ to play the strategy $s_i$, irrespective of the actions 
of the other players and irrespective of which of the points in $K_i$ is her true location. 
 
 \begin{definition}[{minimax dominance}]
  In a mechanism $\mathcal{M} = (X, F)$, an agent $i$ with candidate locations $K_i$ has a minimax dominant strategy $s_i  \in 
\Sigma_i$ if $\forall s'_{i} \in \Sigma_i$ and $\forall s_{-i} \in \Sigma_{-i}$, 
\begin{multline*}
\max_{\ell_i \in K_i} \max_{\sigma_{i} 
\in \Delta(\Sigma_i)} C(\ell_i, F(s_i(K_i), s_{-i}(K_{-i})) - C(\ell_i, F(\sigma_i(K_i), 
s_{-i}(K_{-i})))\\ 
\leq  \max_{\ell_i \in K_i} \max_{\sigma_{i} \in \Delta(\Sigma_i)} 
C(\ell_i, F(s'_i(K_i), s_{-i}(K_{-i})) - C(\ell_i, F(\sigma_i(K_i), s_{-i}(K_{-i})).
 \end{multline*}
 \end{definition}
Before we explain what the definition above implies, let $p = F(s_i(K_i), s_{-i}(K_{-i}))$ be the outcome of the mechanism when agent $i$ plays strategy $s_i$ and all the 
others play some $s_{-i}$. Now, let us consider the term 
\begin{equation} \label{eqn:maxRi}
%  \underbrace{\max_{\ell_i \in K_i} \max_{\sigma_{i} \in \Delta(\Sigma_i)} C(\ell_i, p) - C(\ell_i, F(\sigma_i(K_i), s_{-i}(K_{-i})))}_{\maxRegret_i(p)},
 \maxRegret_i(p) = {\max_{\ell_i \in K_i} \max_{\sigma_{i} \in \Delta(\Sigma_i)} C(\ell_i, p) - C(\ell_i, F(\sigma_i(K_i), s_{-i}(K_{-i})))},
\end{equation}
which calculates agent $i$'s maximum regret (i.e., the absolute worst case loss agent $i$ will experience if and when she realizes her true 
location from her candidate locations) for playing $s_i$ and hence getting the output $p$. Then, what the above definition implies is that for 
a regret minimizing agent $i$ with candidate locations $K_i$, it is always best for $i$ to play $s_i$, irrespective of the actions of the other players, as any other 
strategy $s_i'$ results in an outcome $p'$ with respect to which agent $i$ experiences at least as much maximum regret as she experiences with $p$. 

\textbf{Remark:} Note that both the solution concepts defined above can be seen as natural extensions of the classical (i.e., the usual mechanism design setting where the 
agents know their types exactly) weak dominance notion to our setting. That is, for all $i \in [n]$, if $K_i$ is a single point, then both of 
them collapse to the classical weak dominance notion.
% \footnote{Although very weak dominance is the standard equilibrium solution concept that is used in mechanism design, we note that there is some inconsistency in the 
% literature regarding the usage of the term ``very weak dominance'' as the same notion is sometimes referred to as ``weak dominance.''}.  

As stated in the introduction, given a profile of candidate locations $(K_1, \cdots, K_n)$, we want the mechanism to ``perform well'' against all the possible underlying true 
locations of the agents, i.e., with respect to all the location profiles $I = (\ell_1, \cdots, \ell_n)$ where $\ell_i \in K_i$. Hence, for a given objective function $S$, 
we aim to design mechanisms that achieve a good approximation of the optimal minimax value, which, for $\mathcal{I} = K_1\times\cdots\times K_n$, is denoted by 
$\OMV_S(\mathcal{I})$ and is defined as
\begin{equation}
 \OMV_S(\mathcal{I}) = \maxRegret(p_{opt}, \mathcal{I}),
\end{equation}
where for a point $p \in [0, B]$, if $S(I, p)$ denotes the value of the function $S$ when evaluated with respect to the vector $I$ and a point $p$, then the maximum regret 
associated with $p$ for the instance $\mathcal{I}$ is defined as
\begin{equation} \label{eqn:maxRp}
 \maxRegret(p, \mathcal{I}) =  \max_{I \in \mathcal{I}} \left(S(I, p) - \min_{p' \in {[0,B]}} S(I, p')\right),
\end{equation}
and 
\begin{equation} \label{eq:mr}
 p_{opt} = \argmin_{p \in [0,B]}{\maxRegret(p, \mathcal{I})}. 
\end{equation}
Throughout, we refer to the point $p_{opt}$ as the optimal minimax solution for the instance $\mathcal{I}$. 

Finally, now that we have our performance measure, we define implementation in very weakly dominant and minimax dominant strategies.  

\begin{definition}[Implementation in very weakly dominant (minimax dominant) strategies] \label{imp-vwd-md}
 For a $\delta$-uncertain-facility-location-game, we say that a mechanism $\mathcal{M} = (X, F)$ implements $\alpha$-$\OMV_S$, for some $\alpha \geq 0$ and some 
objective function $S$, in very weakly dominant (minimax dominant) strategies, if for some $s = (s_1, \cdots, s_n)$, where $s_i$ is a very weakly dominant (minimax dominant) 
strategy for agent $i$ with candidate locations $K_i$,
\begin{equation*}
 \maxRegret(F(s_1(K_1), \cdots s_n(K_n)), \mathcal{I}) - \OMV_S(\mathcal{I}) \leq \alpha.
\end{equation*}
\end{definition}

\section{Implementing the average cost objective} \label{sec:avgCost}

In this section we consider the objective of locating a facility so as to minimize the average cost (sometimes succinctly referred to as avgCost and written as AC). While the 
standard objective in the facility location setting is to minimize the sum of costs, here, like in work of \citet{golomb17}, we use 
average cost because since we are approximating additively, it is easy to see that in many cases a deviation 
from the optimal solution results in a factor of order $n$ coming up in the 
approximation bound. Hence, to avoid this, and to make comparisons with our second objective function, maximum cost, easier we use average cost. 

In the standard setting where the agents know their true location, the average cost of locating a facility at a point $p$ is defined as $\frac{1}{n}\sum_{i \in [n]} C(x_i, 
p)$, where $x_i$ is the location of agent $i$. Designing even optimal strategyproof/truthful mechanisms in this case is easy since one can quickly see that the optimal 
location for the facility is the median of $x_1, \cdots, x_n$ and returning the same is strategyproof. Note that, for some $k \geq 0$, when $n = 2k + 1$, the median is 
unique and is the $(k+1)$-th largest element. However, when $n = 2k$, the ``median'' can be defined as any point between (and including) the $(n/2)$-th and 
$((n/2)+1)$-th largest numbers. As a matter of convention, here we consider the $(n/2 + 1)$-th element to be the median. Hence, throughout, we always write that the median 
element is the $(k+1)$-th element, where $k = 
\frac{\lfloor n \rfloor}{2}$. 

In contrast to the standard setting, for some $\delta \in (0, B]$ and a corresponding $\delta$-uncertain-facility-location-game, even computing what the minimax optimal 
solution for the average cost objective (see Equation~\ref{eq:mr}) is is non-trivial, let alone seeing if it can be implemented with any of the solution concepts discussed in 
Section~\ref{sec:solutionConcepts}. Therefore, we start by stating some properties about the minimax optimal solution that will be useful when designing mechanisms. A complete 
discussion on how to find the minimax optimal solution when using the average cost objective, as well as the proofs for the lemmas 
stated in the next section, are in Appendix~\ref{app:mos-ac}. 

\subsection{\texorpdfstring{Properties of the minimax optimal solution for avgCost}{}} \label{sec:opt-AC}
Given the candidate locations $K_i = [a_i, b_i]$ for all $i$, where, for some $\delta \in [0, B]$, $b_i - a_i \leq \delta$, consider the left endpoints associated with all 
the agents, i.e., the set $\{a_i\}_{i \in [n]}$. We denote the sorted order of these points as $L_1, \cdots, L_{n}$ (throughout, by sorted order we mean sorted in 
non-decreasing order). Similarly, we denote the sorted order of the right endpoints, i.e., the points in the set $\{b_i\}_{i \in [n]}$, as $R_1, \cdots, R_{n}$. Next, we 
state the following lemma which gives a succinct formula for the maximum regret associated with a point $p$ (i.e., $\maxRegret(p, \mathcal{I})$, where $\mathcal{I} = [a_1, 
b_1] \times \cdots \times [a_n, b_n]$; see Equation~\ref{eqn:maxRp}). As stated above, all the proofs for lemmas in the section appear in Appendix~\ref{app:mos-ac}. 

\begin{restatable}{lemma}{maxRlemma}\label{clm:maxregret}
Given a point $p$, the maximum regret associated with $p$ for the average cost objective can be written as $\max(obj_1^{AC}(p),\allowbreak obj_2^{AC}(p))$, where
\begin{itemize}
 \item $obj_1^{AC}(p) = \frac{1}{n}\left(2\sum_{i = j}^k (R_i - p) + (n-2k)(R_{k+1} - p)\right)$, where $j$ is the smallest index such that $R_{j} > p$ and $j\leq k$ 
 \item $obj_2^{AC}(p) = \frac{1}{n}\left(2\sum_{i = k+2}^h (p - L_i) + (n-2k) (p - L_{k+1})\right)$, where $h$ is the largest index such that $L_{h} < p$ and $h\geq k+2$.
\end{itemize}
\end{restatable}

Our next lemma states that the minimax optimal solution, $p_{opt}$, associated with the avgCost objective function is always in the interval $[L_{k+1}, R_{k+1}]$. 

\begin{restatable}{lemma}{optInterval}\label{clm:optInterval}
 If $p_{opt}$ is the minimax optimal solution associated with the avgCost objective function, then $p_{opt} \in [L_{k+1}, R_{k+1}]$.
\end{restatable}

Equipped with these properties, we are now ready to talk about implementation using the solution concepts defined in Section~\ref{sec:solutionConcepts}.

\subsection{Implementation in very weakly dominant strategies}
As discussed in Section~\ref{sec:solutionConcepts}, the strongest solution concept that we consider is very weak dominance, where for an agent $i$, with candidate locations 
$K_i$, strategy $s_i$ is very weakly dominant if it is always best for $i$ to play $s_i$, irrespective of the actions of the other players and irrespective of which of the 
points in $K_i$ is her true location. While it is indeed a natural solution concept which extends 
the classical notion of weak dominance, we will see below in, Theorem~\ref{thm:a-vwd}, that it is too strong as no deterministic mechanism can achieve a better 
approximation bound than $\frac{B}{2}$. This in turn implies that, among deterministic mechanisms, the naive mechanism which always, irrespective of the reports of the 
agents, outputs the point $\frac{B}{2}$ is the best one can do.

\begin{theorem} \label{thm:a-vwd}
Given a $\delta \in (0, B]$, let $\mathcal{M}= (X, F)$ be a deterministic mechanism that implements $\alpha$-$\OMV_{AC}$ in very weakly dominant strategies for a 
$\delta$-uncertain-facility-location-game. Then, $\alpha \geq \frac{B}{2}$.
\end{theorem}

\begin{proof}
Let us assume for the sake of contradiction that $\alpha = \frac{B}{2} - \gamma$ for some $\gamma > 0$. 
 First, note that here we can restrict ourselves to direct mechanisms since Chiesa et al.\ showed that the revelation principle holds with respect to this 
solution concept \citep[Lemma A.2]{chiesa14}. So, now, let us consider a scenario where the profile of true candidate locations of the agents are $([a_1, 
b_1],\allowbreak\cdots, [a_n, b_n])$ and let 
$F([a_1, b_1], \allowbreak \cdots, [a_n, b_n]) = p$. Since reporting the true candidate locations is a very weakly dominant strategy in $\mathcal{M}$, this implies that for 
an agent $i$ and for all $\ell \in [a_i, b_i]$,
\begin{equation} \label{eqn:a-vwd1}
 |\ell - p| \leq |\ell - p'|,  
\end{equation} 
where for some $[a'_i, b'_i] \neq [a_i, b_i]$, $p' = F([a_1, b_1], \cdots, [a'_i, b'_i], \allowbreak \cdots, [a_n, b_n])$.

Next, consider the profile of true candidate locations $([a_1, b_1], \allowbreak \cdots, [a'_i, b'_i], \allowbreak \cdots, [a_n, b_n])$. Then, again, using the fact that 
reporting the truth is a very weakly dominant strategy, we have that for agent $i$ and for all $\ell' \in [a'_i, b'_i]$,
\begin{equation} \label{eqn:a-vwd2}
 |\ell' - p'| \leq |\ell' - p|.  
\end{equation}

Equations~\ref{eqn:a-vwd1} and~\ref{eqn:a-vwd2} together imply that for a $k \in [a_i, b_i] \cap [a'_i, b'_i]$, 
\begin{equation} \label{eqn:a-vwd3}
|k - p| =|k - p'|. 
\end{equation} So, if $Q = [a_i, b_i] \cap [a'_i, b'_i]$ and $|Q| > 1$, then Equation~\ref{eqn:a-vwd3} implies that $p = p'$. 

Now,  let us consider $F([a_1, b_1], \cdots, \allowbreak [a_1, b_1])$, where $a_1 = 0$ and $b_1=\epsilon$, and let $F([a_1, b_1], \allowbreak  
\cdots,\allowbreak [a_1, 
b_1]) = p$. By repeatedly using the observation made above, we have that for $\delta \in (0, B]$, {$\epsilon \in \left(0,\min\{\delta, \gamma\}\right)$}, $\epsilon_1 \in 
(0, \epsilon)$, and $b_i = \epsilon + i(\delta - \epsilon_1)$,
\begin{multline} \label{eqn:a-vwd4}
 p = F([a_1, b_1], \cdots, [a_1, b_1]) = F([b_1-\epsilon_1, b_2], \allowbreak [a_1, b_1], \allowbreak \cdots, [a_1, b_1]) \\= F([b_2-\epsilon_1, b_3], [a_1, b_1], 
\cdots, [a_1, b_1]) = \cdots = F([B-\epsilon_1, B], \allowbreak [a_1, b_1], \allowbreak \cdots, [a_1, b_1]) = \cdots\\ = F([B-\epsilon_1, B], \allowbreak \cdots, 
[B-\epsilon_1, 
B]). 
\end{multline} 

Next, it is easy to see that the minimax optimal solution associated with the profile $([a_1, b_1], \allowbreak\cdots, [a_1, b_1])$ is 
$\frac{a_1+b_1}{2} = \frac{\epsilon}{2}$, whereas for the profile $([B-\epsilon_1, B], \cdots, [B-\epsilon_1, B])$ it is $B - \frac{\epsilon_1}{2}$. Also, from 
Equation~\ref{eqn:a-vwd4} we know that $\mathcal{M}$ outputs the same point $p$ for both these profiles. So, if we assume without loss of generality that $p \leq \frac{B + 
\epsilon/2 - \epsilon_{1}/2}{2}$, this implies that for $\mathcal{I} = [B-\epsilon_1, B] \times \cdots \times [B-\epsilon_1, B]$, 
\begin{align*}
 \alpha \geq \maxRegret(p, \mathcal{I})-&\OMV_{AC}(\mathcal{I})\\ &\geq {\regret(p, ( B, \cdots, B))}-{\OMV_{AC}(\mathcal{I})} \\ 
 &\geq {(B/2 - \epsilon/4 + \epsilon_1/4)}-{\epsilon_1/2} \\
 &\geq B/2 -\epsilon\\
 &>  B/2 -\gamma.
\end{align*}
This in turn contradicts our assumption that $\alpha = \frac{B}{2} - \gamma$.
\end{proof}

Although one could argue that this result is somewhat expected given how Chiesa et al.\ also observed similar poor performance for implementation with very weakly dominant
strategies in the context of the single-item auctions \citep[Theorem 1]{chiesa12}, we believe that it is still interesting because not only do we observe a similar 
result in a setting that is considerably different from theirs, but this observation also reinforces their view that one would likely have to look beyond very weakly dominant 
strategies in settings like ours. This brings us to our next section, where we consider an alternative, albeit weaker, but natural, extension to the classical notion 
of weakly dominant strategies. 

\subsection{Implementation in minimax dominant strategies}
In this section we move our focus to implementation in minimax dominant strategies and explore whether by using this weaker solution concept one can obtain a better 
approximation bound than the one obtained in the previous section. To this end, we first present a general result that applies to all mechanisms in our setting that are 
anonymous and minimax dominant, in particular showing that any such mechanism cannot be onto. Following this, we look at non-onto mechanisms and here we provide a mechanism 
that achieves a much better approximation bound than the one we observed when considering implementation in very weak dominant strategies. 

\textbf{Remark:} Note that in this section we focus only on direct mechanisms. This is without loss of generality because, like in the case with 
very weakly dominant strategies, it turns out that the revelation principle holds in our setting for minimax dominant strategies. A proof of the same can be found 
in Appendix~\ref{app:rev-pr}.

\begin{theorem} \label{prop-onto}
 Given a $\delta \in (0, B]$, let $\mathcal{M} = (X, F)$ be a deterministic mechanism that is anonymous and minimax dominant for a $\delta$-uncertain-facility-location-game. 
Then, $\mathcal{M}$ cannot be onto.    
\end{theorem}

\begin{proof}
 Suppose this were not the case and there existed a deterministic mechanism $\mathcal{M}$ that is anonymous, minimax dominant, and onto. First, note that if we restrict 
ourselves to 
profiles where every agent's report is a single point (instead of intervals as in our setting), then $\mathcal{M}$ must have $n-1$ fixed points $y_1 \leq \cdots \leq 
y_{n-1}$ such that for any profile of single reports $(x_1, \cdots, x_n)$, $\mathcal{M}(x_1, \cdots, x_n) = \text{median}(y_1, \cdots, y_{n-1}, x_1, \cdots, x_n)$. This is so 
because, given the fact that $\mathcal{M}$ is anonymous, onto, and minimax dominant in our setting, when restricted to the setting where reports are single points, 
$\mathcal{M}$ is strategyproof, anonymous, and onto, and hence we know from the characterization result by \citet[Corollary 2]{masso11} that every such mechanism must have 
$n-1$ fixed points $y_1 \leq \cdots \leq y_{n-1}$ such that for any profile $(x_1, \cdots, x_n)$, $\mathcal{M}(x_1, \cdots, x_n) = \text{median}(y_1, \cdots, y_{n-1}, p_1, 
\cdots, p_n)$, where $p_i$ is the most preferred alternative of agent $i$ (i.e., agent $i$'s peak; since the utility of agent $i$ for an alternative $a \in [0, B]$ is defined 
as -$|x_i - a|$, we know that the preferences of agent $i$ is symmetric single-peaked with the peak $p_i = x_i$).\footnote{The original statement by \citet[Corollary 
2]{masso11} talks about mechanisms that are anonymous, strategyproof, and efficient. However, it is known that for strategyproof mechanisms in (symmetric) single-peaked 
domains, efficiency is equivalent to being onto (see, e.g., \cite[Lemma 10.1]{nisan07} for a proof).}\footnote{It is worth noting that the characterization result by 
\citet[Corollary 2]{masso11} for mechanisms that are anonymous, stratgeyproof, and onto under symmetric single-peaked preferences is the same as Moulin's characterization of 
the set of such mechanisms on the general single-peaked domain \cite[Theorem 1]{moulin80}).}

Now, given the observation above, for $1\leq j \leq n-1$, let us consider the smallest $j$ such that $y_j \neq y_{j+1}$ (if there is no such $j$ define $j = n-1$ if $y_{n-1} 
\neq B$ and $j = 0$ otherwise) and consider the following input profile $\mathcal{L}_0$
\begin{equation*}
   \left(\underbrace{y_j, \cdots, y_j,}_{n-j-1\text{ agents}} [\ell, r], z, \underbrace{B, \cdots, B}_{j-1 \text{ agents}}\right),  
 \end{equation*}
where $y_j < \ell < r < y_{j+1}$, $r - \ell < \delta$, $z = \frac{\ell + r}{2} - \epsilon$ and $0 < \epsilon < \frac{r-\ell}{2}$.   

In the profile $\mathcal{L}_0$, let $a$ and $b$ denote the agents who report $[\ell, r]$ and $z$, respectively, and let $p_0 = 
\mathcal{M}(\mathcal{L}_0)$. First, note that if $\mathcal{L}_1$ denotes the profile where agent $a$ reports $\ell$ instead of $[\ell, r]$ and every other agent reports as in 
$\mathcal{L}_0$, then $p_1 = \mathcal{M}(\mathcal{L}_1) = \text{median}(y_1, \cdots, \allowbreak y_{n-1},\allowbreak y_j, \cdots,\allowbreak y_j,\allowbreak \ell,\allowbreak 
z, B, \cdots B) = \ell$. Also, if $\mathcal{L}_2$ denotes the profile where agent $a$ reports $r$ instead of $[\ell, r]$ and every other agent reports as in $\mathcal{L}_0$, 
then $p_2 = \mathcal{M}(\mathcal{L}_2) = \text{median} (y_1, \cdots, \allowbreak y_{n-1}, y_j, \cdots, y_j, \allowbreak r, \allowbreak z, \allowbreak B, \cdots B) = z$. Now, 
since $p_1 = \ell$ and $p_2 = z$, this implies that $p_0 = \frac{\ell + z}{2}$, for if otherwise agent $a$ can deviate from $\mathcal{L}_0$ by reporting $\frac{\ell + z}{2}$ 
instead (it is easy to see that this reduces agent $a$'s maximum regret), thus violating the fact that $\mathcal{M}$ is minimax dominant. Given this, consider the profile 
$\mathcal{L}_3$ which is the same as $\mathcal{L}_0$ except for the fact that agent $b$ reports $(2z - \ell)$ instead of $z$. By again using the same line of reasoning as in 
the case of $\mathcal{L}_0$, it is easy to see that $p_3 = \mathcal{M}(\mathcal{L}_3) = z$. However, this in turn implies that agent $b$ can deviate from $\mathcal{L}_0$ to 
$\mathcal{L}_3$, thus again violating the fact that $\mathcal{M}$ is minimax dominant.
\end{proof}

Given the fact that we cannot have an anonymous, minimax dominant, and onto mechanism, the natural question to consider is if we can find non-onto mechanisms that perform 
well. We answer this question in the next section.

\subsubsection{Non-onto mechanisms} \label{sec:mdAC}

\begin{algorithm}[tb]
\centering
\noindent\fbox{%
\begin{varwidth}{\dimexpr\linewidth-4\fboxsep-4\fboxrule\relax}
\begin{algorithmic}[1]
\footnotesize
  \Input a $\delta \geq 0$ and for each agent $i$, their input interval $[a_i, b_i]$ 
  \Output location of the facility $p$
  \State $A \leftarrow \{g_1, g_2, \cdots, g_k\}$, where $g_1 = 0, g_k \leq B$, and $g_{i+1} - g_i = \frac{\delta}{2}$ \label{step:defA}
  \For{each $i \in \{1, \cdots, n\}$}
%     \State $x_i \leftarrow$ point closest to $a_i$ in $A$ {\footnotesize(in case of a tie, break in favour of the point in $[a_i, b_i]$ if $a_i \neq b_i$, break in favour of 
% point to the left otherwise)} \label{step:defx}
%     \State $y_i \leftarrow$ point closest to $b_i$ in $A$ {\footnotesize(in case of a tie, break in favour of the point in $[a_i, b_i]$ if $a_i \neq b_i$, break in favour of 
% point to the left otherwise)}    \label{step:defy}
    \State $x_i \leftarrow$ point closest to $a_i$ in $A$ {\footnotesize(in case of a tie, break in favour of the point in $[a_i, b_i]$ if there exists one, break in favour 
of point to the left otherwise)} \label{step:defx}
    \State $y_i \leftarrow$ point closest to $b_i$ in $A$ {\footnotesize(break ties as in line~\ref{step:defx})}    \label{step:defy}
    \If{$\abs{[x_i, y_i] \cap A} == 1$} \Comment{{\footnotesize the case when $x_i=y_i$}} 
      \State $\ell_i \leftarrow x$ \label{step:abeq}
    \ElsIf{$\abs{[x_i, y_i] \cap A} == 2$}
      \If{$\abs{[x_i, y_i] \cap [a_i, b_i]} < 2$}
	      \If{$a_i + b_i \leq x_i + y_i$} \label{step:ineq}
		\State $\ell_i \leftarrow x_i$ \label{step:ineqassign}
	      \Else
		\State $\ell_i \leftarrow y_i$
	      \EndIf
	\Else \label{step:bothin}
	  \State $\ell_i \leftarrow x_i$ \label{step:bothinassign}
      \EndIf
    \ElsIf{$\abs{[x_i, y_i] \cap A} == 3$}    
      \State $\ell_i \leftarrow z_i$, where $z_i$ is the point in $[x_i, y_i] \cap A$ that is neither $x_i$ nor $y_i$ \label{step:three}
    \EndIf    
  \EndFor 
  \State return $\text{median}(\ell_1, \cdots, \ell_n)$ \label{step:med}   
\end{algorithmic}
\end{varwidth}}
\caption{$\frac{\delta}{2}$-equispaced-median mechanism}
\label{algo:mech}
\end{algorithm}

In this section we consider non-onto mechanisms. We first show a positive result by presenting an anonymous mechanism that implements 
$\frac{3\delta}{4}$-$\OMV_{AC}$ 
in minimax dominant strategies. Following this, we present a conditional lower bound that shows that one cannot achieve an approximation bound better than $\frac{\delta}{2}$ 
when considering mechanisms that have a finite range. 

\paragraph{An anonymous and minimax dominant mechanism.}
Consider the $\frac{\delta}{2}$-equispaced-median mechanism defined in Algorithm~\ref{algo:mech}, which can be thought of as an extension to 
the standard median mechanism. The key assumption in this mechanism is that the designer knows a $\delta$ such that any agent's candidate locations has a length at 
most $\delta$. Given this $\delta$, the key idea is to divide the interval $[0, B]$ into a set of ``grid points'' and then map every profile of reports to one 
of these points, while at the same time ensuring that the mapping is minimax dominant. In particular, in the case of the $\frac{\delta}{2}$-equispaced-median mechanism, 
when $\delta > 0$, its range is restricted to the finite set of points $A = \{g_1, g_2, \cdots g_m\}$ such that, for $i \geq 1$, $g_{i+1} - g_{i} = \frac{\delta}{2}$, $g_1 = 
0$, 
and $g_m \leq B$. 

Below we first prove a lemma where we show that the $\frac{\delta}{2}$-equispaced-median mechanism is minimax dominant. Subsequently, we then 
use this to prove our main theorem which shows that the $\frac{\delta}{2}$-equispaced-median mechanism implements $\frac{3\delta}{4}$-$\OMV_{AC}$ in minimax dominant 
strategies.

\begin{lemma} \label{lem:md}
 Given a $\delta \in [0, B]$ and for every agent $i$ in a $\delta$-uncertain-facility-location-game, reporting the candidate locations $[a_i, b_i]$ is a minimax 
dominant strategy for agent $i$ in the $\frac{\delta}{2}$-equispaced-median mechanism. 
\end{lemma}

\begin{proof}
 Let us fix an agent $i$ and let $[a_i, b_i]$ be her candidate locations. Also, let $\mathcal{L}_0$ be some arbitrary profile of candidate locations, where $\mathcal{L}_0 = 
\left([a_1, b_1], \allowbreak \cdots,\allowbreak [a_n, b_n]\right)$. We need to show that it is minimax dominant for agent $i$ to report $[a_i, b_i]$ in the 
$\frac{\delta}{2}$-equispaced-median mechanism (denoted by $\mathcal{M}$ from now on), and for this we broadly consider the following two cases. Intuitively, in both cases 
what we try to argue is that for an agent $i$ with candidate locations $[a_i, b_i]$ the $\ell_i$ that is associated with $i$ in the mechanism is in fact the agent's ``best 
alternative'' among the alternatives 
in $A$ (see line~\ref{step:defA} in Algorithm~\ref{algo:mech}). 

\textbf{Case 1: $a_i = b_i$.} In this case, we show that it is a very weakly dominant strategy for agent $i$ to report $a_i$. To see this, let $p$ be the 
output of $\mathcal{M}$ when agent $i$ reports $a_i$ and $x_i$ be the point that is closest to $a_i$ in $A$ (with ties broken in favour of the point which is to the left 
of $a_i$). From line~\ref{step:abeq} in the mechanism, we know that $\ell_i = x_i$.  Now, if either $\ell_i < p$ or $\ell_i > p$, then any report that changes the median will 
only result in the output being further away from agent $i$. And if $\ell_i = p$, then since we choose $\ell_i$ to be the point that is closest to $a_i$ in $A$, it is clear 
that it is very weakly dominant for the agent to report $a_i$. Hence, from both the cases above, we have our claim. 

\textbf{Case 2: $a_i \neq b_i$.} Let $p$ be the output of $\mathcal{M}$ when agent $i$ reports $[a_i, b_i]$, and let $x_i$ and $y_i$ be the points that are 
closest (with ties being broken in favour of points in $[a_i, b_i]$ in both cases) to $a_i$ and $b_i$, respectively. From the mechanism we can see that $\ell_i 
\in [x_i, y_i]$. Next, let us first consider the case when $p < x_i$ or $p > y_i$. In both these cases, given the fact that $x_i$ and $y_i$ are the points closest to $a_i$ 
and $b_i$, respectively, $p$ has to be outside $[a_i, b_i]$. And so, if this is the case, then if agent $i$ misreports and the output changes to some $p' < p$ or $p' > p$, in 
both cases, it is easy to see that the maximum regret associated with $p'$ is greater than the one associated with $p$. Hence, the only case 
where an agent $i$ can successfully misreport is if $p \in [x_i, y_i]$. So, we focus on this scenario below. 

Considering the scenario when $p \in [x_i, y_i]$, first, note that the interval $[x_i, y_i]$ can have at most three points that are also in $A$ (this is proved in 
Claim~\ref{clm:atmost3inA}, which is in Appendix~\ref{app:addclms}). So, given this, let us now consider the following cases.

\begin{enumerate}[label=\roman*)]
 \item $\abs{[x_i, y_i] \cap A} = 1$. Since $p \in [x_i, y_i]$, this implies that $p = x_i = \ell_i$. Therefore, in this case, if agent $i$ misreports, then she only 
experiences a greater maximum regret as the resulting output $p'$ would be outside $[x_i, y_i]$ and we know from our discussion above that these points have a greater maximum 
regret than a point in $[x_i, y_i]$. 

\item $\abs{[x_i, y_i] \cap A} = 2$. First, note that since $p \in [x_i, y_i]$ and the only other point in $[x_i, y_i]$ that is also in $A$ is $y_i$, if $p < \ell_i$ or $p > 
\ell_i$, then agent $i$ can only increase her maximum regret by misreporting and changing the outcome (because the new outcome will be outside $[x_i, y_i]$). Therefore, we 
only need to consider the case when $p = \ell_i$, and here we consider the following sub-cases.

\begin{enumerate}
\item $p=\ell_i=x_i$. From the mechanism we know that this happens only when either both $x_i$ and $y_i$ are in $[a_i, b_i]$ 
(lines~\ref{step:bothin}-\ref{step:bothinassign}) or when $a_i + b_i \leq x_i + y_i$ (lines~\ref{step:ineq}-\ref{step:ineqassign}). Now, since we know that every 
point outside $[x_i, y_i]$ is worse in terms of maximum regret than $x_i$ or $y_i$, we only need to consider the case when agent $i$ misreports in such a way that it 
results in the new outcome $p'$ being equal to $y_i$. And below we show that under both the conditions stated above (lines~\ref{step:ineq}-\ref{step:ineqassign} 
and~\ref{step:bothin}-\ref{step:bothinassign}, both of which result in $\ell_i$ being defined as being equal to $x_i$ in the mechanism) the maximum regret associated with 
$y_i$ is at least as much as the one associated with $x_i$. 

To see this, consider the profile where agent $i$ reports $a_i$ instead of $[a_i, b_i]$ and all the other agents' reports are the same as in $\mathcal{L}_0$. Let $p_a$ be 
the output of $\mathcal{M}$ for this profile. Note that $p_a = x_i$ since the $\ell_i$ associated with agent $i$ in this profile is $x_i$ and so the outcome in this profile 
is the same as $p = x_i$. Similarly, let $p_b$ be the outcome when agent $i$ reports $b_i$ instead of $[a_i, b_i]$. Since the $\ell_i$ associated with agent $i$ in this 
profile is $y_i$, one can see that $x_i \leq p_b \leq y_i$. Now, if $p_b = x_i$, then $\maxRegret_i(x_i) = 0$, where $\maxRegret_i(x_i)$ is the maximum regret associated with 
the point $x_i$ for agent $i$ (see Equation~\ref{eqn:maxRi} for the definition), and so $x_i$ is definitely better than $y_i$. Therefore, we 
can ignore this case and instead assume that $p_b = y_i$. So, considering this, since for the maximum regret calculations only the endpoints $a_i$ and $b_i$ matter (this is 
proved in Claim~\ref{maxR@endp}), we have that
\begin{align*}
\maxRegret_i{(y_i)} &= \max\{\abs{a_i - y_i} - \abs{a_i - p_a}, \abs{b_i - y_i} - \abs{b_i - p_b}\}\\
&= \abs{a_i - y_i} - \abs{a_i - p_a} \tag{{\scriptsize since $p_b = y_i$ and $p_a = x_i$}}\\
&\geq \abs{b_i - x_i} - \abs{b_i - p_b} \tag{{\scriptsize depending on the case, use either the fact that  $x_i, y_i \in [a_i, b_i]$ or that $a_i + b_i \leq x_i 
+ y_i$}} \\
&= \maxRegret_i{(x_i)}.  \tag{{\scriptsize since $p_b = y_i$ and $p_a = x_i$}}
\end{align*}

Hence, we see that in this case agent $i$ cannot gain by misreporting.

\item $p = \ell_i = y_i$. We can show that in this case $\maxRegret{(x_i)} \geq \maxRegret{(y_i)}$ by proceeding similarly as in the case above. 
\end{enumerate}
 
 \item $\abs{[x_i, y_i] \cap A} = 3$. Let $x_i < z_i < y_i$ be the three points in $[x_i, y_i] \cap A$. Since the length of $[a_i, b_i]$ is at most $\delta$, note that both 
$x_i$ and $y_i$ cannot be in the interval $(a_i, b_i)$. So below we assume without loss of generality that $x_i \leq a_i$. From the mechanism we know that in this case 
$\ell_i = z_i$ 
(line~\ref{step:three}). Also, as in the cases above, note that if $p < \ell_i$ or $p > \ell_i$, then agent $i$ can only increase her maximum regret by misreporting 
and changing the outcome (because, again, the new outcome will be outside $[x_i, y_i]$). Therefore, the only case we need to consider is if $p = \ell_i$, and for this case we 
show that both $x_i$ and $y_i$ have a maximum regret that is at least as much as that associated with $z_i$ (we do not need to consider points outside $[x_i, y_i]$ since we 
know 
that these points have a worse maximum regret than any of the points in $[x_i, y_i]$). Note that if this is true, then we are done as this shows that agent $i$ cannot 
benefit from misreporting. 

To see why the claim is true, consider the profile where agent $i$ reports $a_i$ instead of $[a_i, b_i]$ and all the other agents' reports are the same as in 
$\mathcal{L}_0$. Let $p_a$ be the outcome of $\mathcal{M}$ for this profile. Similarly, let $p_b$ be the outcome when agent $i$ reports $b_i$ instead of $[a_i, b_i]$. Note 
that $p_a \leq p = z_i$ and $p_b \geq p = z_i$. Now, since, again, for the maximum regret calculations only the endpoints $a_i$ and $b_i$ matter (this is proved in 
Claim~\ref{maxR@endp}, which is in Appendix~\ref{app:addclms}), we have that
\begin{align*}
  \maxRegret(z_i) &= \max\{\abs{a_i-z_i} - \abs{a_i - p_a}, \abs{b_i - z_i} - \abs{b_i - p_b}\} \\
  &\leq \abs{b_i - x_i} - \abs{b_i - p_b} \tag{{\scriptsize since $x_i$ is the closest point to $a_i$ in $A$ and using very weak dominance for single reports}}\\
  &= \max\{\abs{a_i - x_i} - \abs{a_i - p_a}, \abs{b_i - x_i} - \abs{b_i - p_b}\} \tag{{\scriptsize since $x_i \leq a_i$ and $|b_i-p_b| \leq |b_i-p_a|$ by 
 very weak dominance for single reports}}\\
  &= \maxRegret(x_i).
\end{align*}
Similarly, we can show that $\maxRegret(z_i) \leq \maxRegret(y_i)$. Hence, agent $i$ will not derive any benefit from misreporting her candidate locations.  
\end{enumerate}

Finally, combining all the cases above, we have that the $\frac{\delta}{2}$-equispaced-median mechanism is minimax dominant. This concludes the proof of our lemma. 
\end{proof}

Given the lemma above, we can now prove the following theorem. 

\begin{theorem} \label{thm:deltmedian}
 Given a $\delta \in [0, B]$, the $\frac{\delta}{2}$-equispaced-median mechanism is anonymous and implements 
$\frac{3\delta}{4}$-$\OMV_{AC}$ in minimax dominant strategies for a $\delta$-uncertain-facility-location-game.
\end{theorem}

\begin{proof}
 From Lemma~\ref{lem:md} we know that the mechanism is minimax dominant. Therefore, the only thing left to show is that it achieves an approximation bound of 
$\frac{3\delta}{4}$. In order to do this, consider an arbitrary profile of candidate locations $\mathcal{L}_0$, where $\mathcal{L}_0 = ([a_1, b_1], \cdots, [a_n, b_n])$, 
and let $L_1, \cdots, L_n$ and $R_1, \cdots, R_n$ denote the sorted order of the left endpoints (i.e., $\{a_i\}_{i\in [n]}$) and right endpoints  (i.e., $\{b_i\}_{i\in [n]}$), 
respectively. Next, consider the interval $[L_{k+1}, R_{k+1}]$ that is associated with the given profile $\mathcal{L}_0$, where $k = \frac{\lfloor n \rfloor}{2}$ and 
$L_{k+1}$ and $R_{k+1}$ are the medians of the sets $\{a_i\}_{i \in [n]}$ and $\{b_i\}_{i \in [n]}$, respectively. (Note that $R_{k+1} - L_{k+1} \leq \delta$ from 
Claim~\ref{clm:LRdelta}, which is in Appendix~\ref{app:addclms}.) 
Also, let $i$ be an agent who reported an $a_i$ such that $a_i = L_{k+1}$ and $j$ denote one who reported a $b_i$ such that $b_i = R_{k+1}$. Next, from the mechanism, 
consider $x_i$ (see line~\ref{step:defx}) and $y_j$ (see line~\ref{step:defy}). 

First, for every agent $c$ with $b_c \leq R_{k+1}$ (and there are $k+1$ of them), $\ell_c \leq y_j$. This is so because, for every agent $c$, the point 
$\ell_c$ that is associated with the agent is in $[x_c, y_c]$ and also for two agents $c, c'$, if $b_{c'} \leq b_c$, then $y_{c'} \leq y_c$. Similarly, for every agent $d$ 
with $a_d \geq L_{k+1}$ (and, again, there are $k+1$ of them), $\ell_d \geq x_i$. Hence, it follows from the two observations above that the output $p$ of 
the mechanism, which is the median of $\{\ell_1, \cdots, \ell_n\}$, will be in the interval $[x_i, y_j]$.  

Second, we will show that if $\abs{[x_i, y_j] \cap A} = 3$, then the output $p$ of the mechanism will be the point $z$ in interval $(x_i, y_j)$ such that $z \in A$. To 
see this, consider the points $R_1, \cdots, R_{k+1}$ and consider the largest $q$, where $q \leq k$, such that $R_q \leq z$ (define $q=0$, if $R_1 > z$). We will assume 
without of loss of generality that, for each $i$, the agent associated with $R_i$ (i.e., one who reports a right endpoint such that it is equal to $R_i$) is agent $i$. Now, 
for each agent $r$ from 1 to $q$, the $\ell_r$ 
associated with them in the mechanism is at most $z$. This is so because, for each such agent $r$ the $y_r$ associated with them at most $z$. Also, for 
each agent $s$ from $q+1$ to $k+1$, the only way the $\ell_s$ associated with agent $s$ is greater than $z$ is if the left endpoints associated  with them are greater than 
$L_{k+1}$ (because if not, then one can see that the $[x_s, y_s]$ associated with agent $s$ has $x_i$ and $z$ in it and so $\ell_s$ cannot be greater than $z$). Now, let 
$n_1$ be the number of agents among the agents from $q+1$ to $k+1$ such that their left endpoints are 
greater than $L_{k+1}$. This implies that there are $n_1$ agents among the ones that report a left endpoint in $\{L_1, \cdots, L_{k+1}\}$ that have a corresponding right 
endpoint greater than or equal to $R_{k+1}$. And this in turn implies that the $\ell$s associated with these $n_1$ agents can be at most $z$ (because, again, one can see 
that 
the $[x, y]$ associated with such an agent will have $x_i$ and $z$ in it). Hence, combining all the observations above, we see that $q + ((k+1 - (q+1) + 1) - n_1) + n_1 = 
k+1$ agents report an $\ell$ that is at most $z$. Now, we can employ a similar line of reasoning to show that $k+1$ agents report an $\ell$ that is at least $z$. Hence, it 
follows that the median of $\{\ell_1, \cdots, \ell_n\}$ must be $z$. 

Given the observations above, let us now look at the following cases. In all the cases we show that any point in $[L_{k+1}, R_{k+1}]$ is at a distance of at most 
$\frac{3\delta}{4}$ from $p$, the output of the mechanism. 

\textbf{Case 1:} both $x_i$ and $y_j$ are not in $[L_{k+1}, R_{k+1}]$. In this case, let us consider the following sub cases. 
\begin{enumerate}[label=(\alph*)]
 \item[a)] $\abs{[x_i, y_j] \cap A} = 2$. Since both $x_i$ and $y_j$ are not in $[L_{k+1}, R_{k+1}]$ and there is no other point in $A$ that is in $[L_{k+1}, R_{k+1}]$, 
we have that the distance from any point in $[L_{k+1}, R_{k+1}]$ to either $x_i$ or $y_j$ is at most $\frac{\delta}{2}$ (as the distance between $x_i$ or $y_j$ is 
$\frac{\delta}{2})$. Also, from the discussion above we know that the median that is returned by the mechanism is either $x_i$ or $y_j$, so we have our bound.   

\item[b)] $\abs{[x_i, y_j] \cap A} = 3$. In this case we know from above that the output of the mechanism $p = z$, where $z \in [x_i, y_j] \cap A$ and is neither $x_i$ nor 
$y_j$. Therefore, we have that $(z - L_{k+1}) \leq (z- x_i) = \frac{\delta}{2}$, and $(R_{k+1} - z) < (y_j - z) = \frac{\delta}{2}$. Hence, any point in $[L_{k+1}, 
R_{k+1}]$ is a distance of at most $\frac{\delta}{2}$ from $p$. 

\item[c)] $\abs{[x_i, y_j] \cap A} > 3$. Since for any interval $[a, b]$ of length at most $\delta$, the $[x, y]$ associated with it can have at 
most three points in $A$ 
(this is proved in Claim~\ref{clm:atmost3inA}, which is in Appendix~\ref{app:addclms}), where $x$ and $y$ are as defined in the mechanism (see 
lines~\ref{step:defx}-\ref{step:defy}), this case is impossible. 
\end{enumerate}

\textbf{Case 2:} at least one of $x_i$ or $y_j$ is in $[L_{k+1}, R_{k+1}]$. Let us assume without loss of generality that $x_i$ is the point in $[L_{k+1}, R_{k+1}]$, and 
let 
us consider the following sub cases.
\begin{enumerate}
 \item[a)] $\abs{[x_i, y_j] \cap A} = 2$. In this case, if the output of the mechanism $p = x_i$, then we have that $(x_i - L_{k+1}) \leq \frac{\delta}{4}$ (since, by 
definition, $x_i$ is the point in $A$ that is closest to $L_{k+1}$), and $(R_{k+1} - x_i) \leq |R_{k+1} - y_j| + (y_j - x_i) \leq \frac{\delta}{4} + \frac{\delta}{2}$ 
(since, 
by definition, $y_j$ is the point in $A$ that is closest to $R_{k+1}$). Hence, any point in $[L_{k+1}, R_{k+1}]$ is a 
distance of at most $\frac{3\delta}{4}$ from $p$. On other hand, if the output of the mechanism $p = y_j$, then we have that $(y_j - L_{k+1}) = (y_j - x_i + x_i - L_{k+1}) 
\leq \frac{\delta}{2} + \frac{\delta}{4}$ (since $x_i$ is the point in $A$ that is closest to $L_{k+1}$), and $|y_j - R_{k+1}| \leq \frac{\delta}{4}$ (since $y_j$ is the 
point in $A$ that is closest to $R_{k+1}$). Hence, when  $p = y_j$, any point in $[L_{k+1}, R_{k+1}]$ is a distance of at most $\frac{3\delta}{4}$ from $p$. Combining the 
two, we see that our claim is true in this case. 

\item[b)] $\abs{[x_i, y_j] \cap A} = 3$. Here again since $p = z$, where $z \in [x_i, y_j] \cap A$ and is neither $x_i$ nor 
$y_j$, we have that $(z - L_{k+1}) \leq (z- x_i + x_i -  L_{k+1}) \leq \frac{\delta}{2} + 
\frac{\delta}{4}$ (since $x_i$ is the point in $A$ that is closest to $L_{k+1}$), and $(R_{k+1} - z) \leq (y_j - z) \leq \frac{\delta}{2}$. Hence, even in this case, any 
point in $[L_{k+1}, R_{k+1}]$ is at a distance of at most $\frac{3\delta}{4}$ from $p$. 
\end{enumerate}

Finally, since the output of the mechanism is a distance of at most $\frac{3\delta}{4}$ from any point in $[L_{k+1}, R_{k+1}]$, and given the fact that the minimax optimal  
solution 
is in $[L_{k+1}, R_{k+1}]$ (from Lemma~\ref{clm:optInterval}), one can see using Lemma~\ref{clm:maxregret} that for $\mathcal{I} = [a_1, b_1] \times \cdots  \times [a_n, 
b_n]$,  $\maxRegret(p, \mathcal{I}) - \text{OMV}_{AC}(\mathcal{I})$ is also bounded by $\frac{3\delta}{4}$ (this is proved in Claim~\ref{clm:mr-ac-bounds}, which is in 
Appendix~\ref{app:addclms}). 
\end{proof}

\paragraph{A conditional lower bound.}
In the context of our motivating example from the introduction, it is possible, and in fact quite likely, that the city can only build the school at a finite set of locations 
on the street. Therefore, an interesting class of non-onto mechanisms to consider is ones which have a finite range. Furthermore, seeing our mechanism above, an inquisitive 
reader 
might wonder: ``why $\frac{\delta}{2}$-equispaced? why not $\frac{\delta}{3}$-equispaced or something smaller than $\frac{\delta}{2}$?'' First, one can easily construct 
counter-examples to show that any $\epsilon$-equispaced-median mechanism is not minimax dominant for $\epsilon < \frac{\delta}{2}$. However, that still does not rule out 
mechanisms whose range is some finite set $\{g_1, \cdots, g_m\}$. So, below we consider this question and we show that the approximation bound associated with any mechanism 
that is anonymous, minimax dominant, and has a finite range, is at least $\frac{\delta}{2}$. The key idea that is required in order to show this bound is the following lemma, 
which informally says that if the mechanism has a finite range, is minimax dominant, and achieves a bound less than $\frac{3\delta}{4}$, then there is ``sufficient-gap'' 
between 
four consecutive points in the range, $A$, of the mechanism. Once we have this observation it is then in turn used to construct profiles that will result in 
the 
stated bound.  

\begin{lemma} \label{lem:finite-lem}
Given a $\delta \in (0, \frac{B}{6}]$, let $\mathcal{M}$ be a deterministic mechanism that has a finite range $A$ (of size at least six), is anonymous, and one 
that implements $\alpha$-$\OMV_{AC}$ in minimax dominant strategies for a $\delta$-uncertain-facility-location-game.
Then, either $\alpha \geq \frac{3\delta}{4}$, or there exists four consecutive points $g_1, g_2, g_3, g_4 \in A $ such that $g_1 < g_2 < g_3 < g_4$ and 
$\frac{d_1}{2} + d_2 + \frac{d_3}{2} \geq \delta$, where, for $i \in [3], d_i = g_{i+1} - g_i$.
\end{lemma}

\begin{proof}[Proof (sketch)]
The proof here is broadly similar to the proof of Theorem~\ref{prop-onto} in that here again we use a characterization result by 
\citet{masso11}, albeit a different one, and then construct profiles in order to prove that the claim is true. The result we will rely on is the one that characterizes the 
set of anonymous and strategyproof mechanisms in the symmetric single-peaked domain \cite[Corollary 1]{masso11}. In particular, \citeauthor{masso11} showed that any mechanism 
that is anonymous and strategyproof in the symmetric single-peaked domain can be described as a `disturbed median'. We will not be defining disturbed medians precisely since 
the definition is quite involved and we do not need it for the purposes of this paper, but broadly a mechanism $\mathcal{M'}$ is a disturbed median if i) it has $n+1$ fixed 
points $0 \leq y_1 \leq \cdots \leq y_{n+1} \leq 1$, ii) its range has a countable number of non-intersecting discontinuity intervals $\{[a_m, b_m]\}_{m \in M}$, where $M$ is 
some indexation set and where for all $i$, $y_i \notin [a_m, b_m]$ for all $m \in M$, iii) has a family of anonymous tie-breaking rules, and iv) behaves the 
following way with respect to a profile $(x_1, \cdots, x_n)$ of exact reports: if $\text{median}(y_1, \cdots, y_{n+1}, x_1, \cdots, x_n) \neq d_m$ for all $m\in M$, where 
$d_m = \frac{a_m+b_m}{2}$, then $\mathcal{M'}(x_1, \cdots, x_n) = \text{median}(y_1, \cdots, y_{n+1}, t_1, \cdots, t_n)$, where $t_i$ is the most preferred alternative of 
agent $i$ in the range of $\mathcal{M'}$.\footnote{Note that we have only defined it for profiles which satisfy the property mentioned above since we will only be using such 
profiles in this proof; the interested reader can refer to \cite[Definition 7]{masso11} for a complete definition of disturbed medians.}

Given the partial description of disturbed medians mentioned above, let us consider the mechanism $\mathcal{M}$ that is anonymous and minimax dominant. First, note that if we 
restrict ourselves to profiles where every agent's report is a single point (instead of intervals as in our setting), then we know from the result of 
\citeauthor{masso11} that $\mathcal{M}$ must be a disturbed median \cite[Corollary 1]{masso11}. Second, note that for each $y_i$, $\mathcal{M}(y_i, \cdots, y_i) = 
\text{median}(y_1, \cdots, \allowbreak y_{n+1}, \allowbreak y_i, \cdots, y_i) = y_i$. Therefore, each $y_i$ belongs to $A$, the range of $\mathcal{M}$. Additionally, also note 
that $y_1$ and $y_{n+1}$ are the minimum and maximum elements, respectively, in $A$, for if not then for a $g \in A$ that is either less than $y_1$ or greater than $y_{n+1}$, 
$\text{median}(y_1, \cdots, y_{n+1}, \allowbreak g, \cdots, g) \neq g$, which is impossible if $g \in A$ and $\mathcal{M}$ is minimax dominant. 

Next, given the observations above, consider the following two cases.

\textbf{Case 1: there exists at most 2 points in $A$ that are greater than $y_2$.} In this case, let us consider the points $y_1$ and $y_2$. If $y_2 - y_1 \geq 
\frac{3\delta}{2}$, then consider the profile $\mathcal{L}_0$ where $n-1$ agents report $y_1$ and the last agent reports $y_2$, i.e., $\mathcal{L}_0 = (y_1, \cdots, y_1, 
y_2)$. Let $p_0 = \mathcal{M}(\mathcal{L}_0)$. From above we know that $p_0 = \text{median}(y_1, \cdots, y_{n+1}, \allowbreak y_1, \cdots, y_1, y_2) = y_2$ (since $y_2$ is 
the $(n+1)$th largest number). However, the minimax optimal solution associated with this profile is $y_1$ and now one can verify from the expressions in 
Lemma~\ref{clm:maxregret} that, for an appropriate choice of $n$, $\alpha\geq\maxRegret(p_0, \mathcal{I}) - \OMV_{AC}(\mathcal{I}) \geq \frac{3\delta}{4}$, where 
$\mathcal{I}$ is the instance associated with 
the profile $\mathcal{L}_0$. Hence, ``the either part'' of our lemma is true in this case. 

On the other hand, if $y_2 - y_1 < \frac{3\delta}{2}$, then consider the case when there are two points $g, g' \in A$ such that $y_2 < g' < g$ and the three points are 
consecutive in $A$ (the arguments below can be easily modified when there is only one or zero points greater than $y_2$). Note that if $g' -y_2 \geq \frac{3\delta}{2}$ or $g - 
g' \geq 
\frac{3\delta}{2}$, then we can construct profiles like in the case above to show that $\alpha \geq \frac{3\delta}{4}$. Therefore, let us assume that 
this is not the case. Now, this implies that if $d_1 = (y_1 - 0)$ and $d_2 = (B - g)$, then $B-0 = B - g + g - y_2 + y_2 - y_1 + y_1 - 0 < d_2 + \frac{3\delta}{2} + 
\frac{3\delta}{2} + \frac{3\delta}{2} + d_1$. However, this in 
turn implies that at least one of $d_1$ or $d_2$ is at least $\frac{B}{2} - \frac{9\delta}{4} \geq \frac{3\delta}{4}$. So, if we assume without loss of generality that 
$d_1 \geq \frac{3\delta}{4}$, then we can consider the profile $\mathcal{L}_0$ where all the agents report 0 (note that by our assumption on $d_1$, $0 \notin A$, because if 
so, then since $y_1$ is the minimum element in $A$, $y_1$ should be equal to 0), and let $p_1 = \mathcal{M}(\mathcal{L}_1)$. Given the fact that $p_1 \geq y_1$, this in turn 
implies that one can verify from the expressions in Lemma~\ref{clm:maxregret} that $\alpha\geq\maxRegret(p_1, \mathcal{I}) - \OMV_{AC}(\mathcal{I}) \geq \frac{3\delta}{4}$, 
where $\mathcal{I}$ is the instance associated with the profile $\mathcal{L}_1$. Hence, again, ``the either part'' of our lemma is true.     

\textbf{Case 2: there are at least 3 points in $A$ that are greater than $y_2$.} In this case, consider four consecutive points $g_1, g_2, g_3, g_4$ in $A$ such that $g_1 
\geq y_2$ and $g_1 < g_2 < g_3 < g_4$. Also, 
for $i\in[3]$, let $d_i$ denote the distance between the points $g_i$ and $g_{i+1}$. Below we will show that $\frac{d_1}{2} + d_2 + \frac{d_3}{2} \geq \delta$.

Now, to prove our claim, let us assume for the sake of contradiction that $\frac{d_1}{2} + d_2 + \frac{d_3}{2} < \delta$. Also, let us consider the 
largest $j$ such that $y_j \leq g_1$ (note that we have $2\leq j \leq n$), and let us consider the following input profile $\mathcal{L}_0$
\begin{equation*}
   \left(\underbrace{y_j, \cdots, y_j,}_{n-j\text{ agents}} [\ell, r], g_3, \underbrace{B, \cdots, B}_{j-2 \text{ agents}}\right),  
 \end{equation*}
 where $\ell = (g_1 + \frac{d_1}{2} - \gamma_1)$, $r = (g_3 + \frac{d_3}{2} + \gamma_2)$, $\gamma_1 + \gamma_2 < (\delta - (\frac{d_1}{2} + x_2 + \frac{d_3}{2}))$, $\gamma_1 
< \gamma_2$, $\gamma_1 < \frac{d_1}{2}$, and $\gamma_2 < \frac{d_3}{2}$.     

Let $p_0 = \mathcal{M}(\mathcal{L}_0)$ and let $a$ and $b$ be the agents who report $[\ell, r]$ and $g_3$, respectively, in $\mathcal{L}_0$. Next, consider the profile 
$\mathcal{L}_1$ where the only difference from $\mathcal{L}_0$ is that agent $a$ reports $\ell$ here instead of $[\ell, r]$, and let $p_1 = \mathcal{M}(\mathcal{L}_1)$. Now, 
given 
the fact that $g_1, g_2, g_3, g_4$ are consecutive points in $A$ and $g_1$ is the unique closest point to $\ell$ in $A$, we know that $g_1$ is the most preferred 
alternative of agent $a$ in the range of $\mathcal{M}$. Therefore, we have that $p_1 = \text{median}(y_1, \cdots, y_{n+1},\allowbreak y_j, \cdots,\allowbreak 
y_j,\allowbreak g_1,\allowbreak g_3, B, \cdots B) = g_1$ (since $g_1$ is the $(n+1)$th largest number). Using a similar line of reasoning one can see that if 
$\mathcal{L}_2$ denotes 
the profile where agent $a$ reports $r$ instead of $[\ell, r]$  and every other agent reports as in $\mathcal{L}_0$, then $p_2 = \mathcal{M}(\mathcal{L}_2) = 
\text{median}(y_1, \cdots, y_{n+1}, y_j, \cdots, y_j, g_4, g_3, B, \cdots B) = g_3$ (since here $g_3$ is the $(n+1)$th largest number). Now, since $p_1 = g_1$ and $p_2 = 
g_3$, this 
implies that $p_0 = g_2$, for if otherwise agent $a$ can deviate from $\mathcal{L}_0$ by reporting $g_2$ instead (it is easy to see that this reduces agent $a$'s maximum 
regret), thus violating the fact that $\mathcal{M}$ is minimax dominant. Given this, now consider the profile $\mathcal{L}_3$ which is the same as $\mathcal{L}_0$ except for 
the 
fact that agent $b$ reports $g_4$ instead of $g_3$. Here $p_3 = \mathcal{M}(\mathcal{L}_3)$ has to be equal to $g_3$ because one can see from the max.\ regret calculations 
associated with agent $a$ that $a$ has a lesser maximum regret for $g_3$ than for $g_2$ (this is because of our choice of appropriate $\gamma_1$ and $\gamma_2$). So 
if $p_3$ is not equal to $g_3$, then agent $a$ can move from $\mathcal{L}_3$ to $\mathcal{L}_4$ where $\mathcal{L}_4$ is the same as $\mathcal{L}_3$ except for agent $a$ 
reporting 
$g_3$ instead of $[\ell, r]$ (and this would be beneficial for $a$ since the output for $\mathcal{L}_4 = \text{median}(y_1, \cdots, 
y_{n+1},\allowbreak y_j, \cdots,\allowbreak y_j,\allowbreak g_3,\allowbreak g_4, B, \cdots B) = g_3$). However, now if $p_3 = 
g_3$, then this implies that agent $b$ can deviate from $\mathcal{L}_0$ (which, as discussed above, has an output $p_0 = g_2$) to $\mathcal{L}_3$, thus again violating the 
fact that $\mathcal{M}$ is minimax dominant. Hence, for this case, we have that $\frac{d_1}{2} + d_2 + \frac{d_3}{2} \geq \delta$.

Finally, combining all the cases above, we have our lemma. 
\end{proof}

Given the lemma above, the proof of our lower bound is straightforward. (Note that below we ignore mechanisms which have less than six points in their range 
as one can easily show that such mechanisms perform poorly.)

\begin{theorem} \label{thm:lb-ac}
  Given a $\delta \in (0, \frac{B}{6}]$, let $\mathcal{M}$ be a deterministic mechanism that has a finite range (of size at least six), is anonymous, and one that 
implements $\alpha$-$\OMV_{AC}$ in minimax dominant strategies for a $\delta$-uncertain-facility-location-game. Then, for any $\epsilon > 0$, $\alpha 
\geq \frac{\delta}{2} - \epsilon$. 
\end{theorem}

% \begin{proof}
% Consider the mechanism $\mathcal{M}$ and let $A$ denote its range. From Lemma~\ref{lem:finite-lem} we know that either $\alpha \geq \frac{\delta}{2}$ or there exists $g_1, 
% g_2 \in A \cup \{0, B\}$ such that $g_2 - g_1 \geq \frac{\delta}{2}$ and $\abs{[g_1, g_2] \cap \{A \cup \{0, B\}\}} = 2$. Since the former case results in a bound that is 
% bigger than the one in the statement of our theorem, below we just consider the latter case where there exists $g_1, g_2$ satisfying the conditions stated above. 
% 
% To show our bound, consider the profile $\mathcal{L}_0$ where $\mathcal{L}_0 = (\frac{g_1 + g_2}{2}, \cdots, \frac{g_1 + g_2}{2})$, and let $p_0 = 
% \mathcal{M}(\mathcal{L}_0)$. Next, since $\abs{[g_1, g_2] \cap \{A \cup \{0, B\}\}} = 2$, we know that $p_0 \geq g_2$ or $p_0 \leq g_1$. Therefore, let us assume without loss 
% of generality that $p_0 \leq g_1$. Now, notice that the minimax optimal solution, $p_{opt}$, for this profile is $\frac{g_1 + g_2}{2}$. And so, using the fact that $g_2 - g_1 
% \geq \frac{\delta}{2}$ and that $p_{opt} - p_0 \geq \frac{g_1+g_2}{2} - g_1 \geq \frac{\delta}{4}$, we can now use the expressions from Lemma~\ref{clm:maxregret} to 
% see that $\maxRegret(p_0, \mathcal{I}) - \OMV_{AC}(\mathcal{I}) \geq \frac{\delta}{4}$, where $\mathcal{I}$ is the instance associated with the profile $\mathcal{L}_0$. This 
% in turn gives us our lower bound.
% \end{proof}

\begin{proof}[Proof (sketch)]

Consider the mechanism $\mathcal{M}$ and let $A$ denote its range. From Lemma~\ref{lem:finite-lem} we know that either $\alpha \geq \frac{3\delta}{4}$ or there exists four 
consecutive points $g_1, g_2, g_3, g_4 \in A $ such that $g_1 < g_2 < g_3 < g_4$ and $\frac{d_1}{2} + d_2 + \frac{d_3}{2} \geq \delta$, where, for $i \in [3], d_i = g_{i+1} 
- g_i$. Since the former case results in a bound that is bigger than the one in the statement of our theorem, below we just consider the latter case where there exists $g_1$, 
$g_2, g_3, g_4$ satisfying the conditions stated above. Also, since $\mathcal{M}$ is minimax dominant and anonymous, we can again make use of \citeauthor{masso11}'s 
characterization result~\cite[Corollary 1]{masso11} as we did in Lemma~\ref{lem:finite-lem}.  

Now, since $\frac{d_1}{2} + d_2 + \frac{d_3}{2} \geq \delta$, we know that at least one of $d_1, d_2$ or $d_3$ is at least $\frac{\delta}{2}$. Also if any of $d_1, d_2$, 
or $d_3$ is at least $\delta$, then it is easy to construct profiles so as to achieve our bounds. Therefore, for the rest of the proof we assume that $d_c < \delta$, for all 
$c \in [3]$. So, given this, let $d_i \geq \frac{\delta}{2}$ for some $i \in [3]$, and let us consider the largest $j$ such that $y_j \leq g_i$ (again, like in the proof 
of Lemma~\ref{lem:finite-lem}, $j\leq n$). We have the following two cases, where $k = \lfloor \frac{n}{2} \rfloor$.

\textbf{Case 1: $j \geq n- k$.} Consider the profile $\mathcal{L}_0$ where $k+1$ agents report $g_i$ and the rest of the agents report $g_{i+1}$.  
Let $p_0 = \mathcal{M}(\mathcal{L}_0)$. Since $p_0 = \text{median}(y_1, \cdots, y_{n+1},\allowbreak g_{i}, \cdots,\allowbreak g_{i},\allowbreak g_{i+1}, \cdots g_{i+1})$, and 
since $j \geq n-k$, we have that $p_0 = g_i$.  Next, consider the profile $\mathcal{L}_1$, where the only change from $\mathcal{L}_0$ is that here agent 1 reports $[g_i, 
g_{i+1} - \gamma]$, where $0 < \gamma < \min(\frac{d_{i}}{2}, \frac{\delta}{2})$, instead of $g_i$. Let $p_1 = \mathcal{M}(\mathcal{L}_1)$. Now, if either $p_1 > g_i$ or $p_1 
\leq g_{i-1}$ (if such a point exists in $A$), it is easy to see that agent 1 will deviate to $\mathcal{L}_0$ since her maximum regret for the point $g_i$ is lesser in 
either of the cases. This in turn implies that $p_1 = g_i$. Continuing this way one can reason along the same lines that for $c \leq k+1$ the profile 
$\mathcal{L}_c$, where $\mathcal{L}_c$ is the same as $\mathcal{L}_{c-1}$ except for agent $c$ reporting $[g_i, g_{i+1} - \gamma]$ instead of $g_i$, the output $p_c$ 
associated with $\mathcal{L}_c$ is equal to $g_i$. 

Given the observations above, consider the profile $\mathcal{L}_{k+1}$ where the first $k+1$ agents report $[g_i, g_{i+1} - \gamma]$ and the rest of the agents report 
$g_{i+1}$. Now,
one can calculate using Algorithm~\ref{algo1} that the minimax optimal solution $p_{opt} \geq \frac{g_i+(2k+1)(g_{i+1} - \gamma)}{2(k+1)}$. Also, we know from above 
that $p_{k+1} = \mathcal{M}(\mathcal{L}_{k+1}) = g_i$. And so, using the fact that $d_i = g_{i+1} - g_i \geq \frac{\delta}{2}$ and that $p_{opt} - p_q \geq 
\frac{(2k+1)(g_{i+1} - 
g_i - \gamma)}{2(k+1)}$ we can now use the expressions from Lemma~\ref{clm:maxregret} to see that $\maxRegret(p_{k+1}, \mathcal{I}) - \OMV_{AC}(\mathcal{I}) \geq 
\frac{(2k+1)(g_{i+1} - g_i - \gamma)}{2(k+1)}$, which for an appropriately chosen value of $n$ and $\gamma$ is greater than or equal to $\frac{\delta}{2} - \epsilon$ for any 
$\epsilon > 0$. 

 \textbf{Case 2: $j < n- k$.} We can handle this similarly as in the previous case. In particular, consider the profile $\mathcal{L}_0$ where $k+1$ agents report $g_i$ 
and the rest of the agents report $g_{i+1}$. Let $p_0 = \mathcal{M}(\mathcal{L}_0)$. Since $p_0 = \text{median}(y_1, \cdots, y_{n+1},\allowbreak g_{i}, \cdots,\allowbreak 
g_{i},\allowbreak g_{i+1}, \cdots g_{i+1})$, and since $j < n-k$, we know that $p_0 = \min(g_{i+1}, y_{j+1})$ (as the (n+1)th largest number will be either $y_{j+1}$ or 
$g_{i+1}$). However, we know that $y_{j+1}$ is in the range of $A$ since $\mathcal{M}(y_{j+1}, \cdots, y_{j+1}) = \text{median}(y_1, \cdots, \allowbreak y_{n+1}, \allowbreak 
y_{j+1}, \cdots, y_{j+1}) = y_{j+1}$, and also that $g_i$ and $g_{i+1}$ are consecutive in $A$. Therefore, $y_{j+1} \geq g_{i+1}$, and we have that $p_0 = g_{i+1}$.  Next, 
consider the profile $\mathcal{L}_1$, where the only change from $\mathcal{L}_0$ is that here agent $(k+2)$ reports $[g_i + \gamma, g_{i+1}]$, where $0 < \gamma < 
\min(\frac{d_{i}}{2}, \frac{\delta}{2})$, instead of $g_{i+1}$. Let $p_1 = \mathcal{M}(\mathcal{L}_1)$. Now, if either $p_1 > g_{i+1}$ (if such a point exists in $A$) or $p_1 
\leq g_{i}$, it is easy to see that agent $(k+2)$ will deviate to $\mathcal{L}_0$ since her maximum regret for the point $g_{i+1}$ is lesser in 
either of the cases. This in turn implies that $p_1 = g_{i+2}$. Continuing this way one can reason along the same lines that for $c \geq k+2$ the profile 
$\mathcal{L}_{c-k-1}$, where $\mathcal{L}_{c-k-1}$ is the same as $\mathcal{L}_{c-k-2}$ except for agent $c$ reporting $[g_i + \gamma, g_{i+1}]$ instead of $g_{i+1}$, the 
output $p_c$ associated with $\mathcal{L}_c$ is equal to $g_{i+1}$.  
 
Given the observations above, consider the profile $\mathcal{L}_{n-k-1}$ where the last $n-k-1$ agents report $[g_i+\gamma, g_{i+1}]$ and the rest of the agents report 
$g_i$. Now, one can calculate using Algorithm~\ref{algo1} that the minimax optimal solution $p_{opt} = g_i$. Also, we know from above 
that $p_{n-k-1} = \mathcal{M}(\mathcal{L}_{n-k-1}) = g_{i+1}$. And so, using the fact that $d_i = g_{i+1} - g_i \geq \frac{\delta}{2}$, we can now use the expressions from 
Lemma~\ref{clm:maxregret} to see that $\maxRegret(p_{n-k-1}, \mathcal{I}) - \OMV_{AC}(\mathcal{I}) \geq \frac{(2k-2)(g_{i+1} - g_i - \gamma)}{2k}$, which for an 
appropriately chosen value of $n$ and $\gamma$ is greater than or equal to $\frac{\delta}{2} - \epsilon$ for any $\epsilon > 0$. 

Finally, combining the two cases above, we have our lower bound.
\end{proof}

\section{Implementing the maximum cost objective} \label{sec:maxCost}
In this section we turn our attention to the objective of minimizing the maximum cost (sometimes succinctly referred to as maxCost and written as MC) which is another 
well-studied 
objective function in the context of the facility location problem. In the standard setting where the reports are exact, the maximum cost associated with locating a facility 
at $p$ is defined as $\max_{i \in [n]} C(x_i, p)$ and if we assume without loss of generality that the $x_i$s are in sorted order, then one can easily see that the optimal 
solution to this objective is to locate the facility at $p = \frac{x_1 + x_n}{2}$. However, unlike in the case of the average cost objective 
that was considered in Section~\ref{sec:avgCost}, one cannot design an optimal strategyproof mechanism even when the reports are exact, and it is known that the 
best one can do in terms of additive approximation is to achieve a bound of $\frac{B}{4}$ in the case of deterministic mechanisms and $\frac{B}{6}$ in the case of randomized 
mechanisms \citep[Theorems 5, 15]{golomb17}. 

Now, coming to our setting, unlike in the case of the average cost objective, calculating the minimax optimal solution is straightforward in this case. In fact, given the 
candidate locations $[a_i, b_i]$ for all $i$, if $L_1, \cdots, L_{n}$ and $R_1, \cdots, R_n$ denote the sorted order of the points in $\{a_i\}_{i \in [n]}$ and $\{b_i\}_{i 
\in [n]}$, respectively, then it is not too hard to show that the minimax optimal solution is the point $\frac{L_1 + R_1 + L_n + R_n}{4}$ (a complete discussion on 
how to find the minimax optimal solution when using the maximum cost objective is in Appendix~\ref{app:mos-mc}). Therefore, below we 
directly move on to implementation using the solution concepts defined in Section~\ref{sec:solutionConcepts}.  

\subsection{Implementation in very weakly dominant strategies}
In the case of the maximum cost objective we again see that very weak dominance is too strong a solution concept as even here it turns out that we cannot do any better than 
the naive mechanism which always outputs the point $\frac{B}{2}$ as the solution. The following theorem, which can be proved by proceeding exactly like in the proof of 
Theorem~\ref{thm:a-vwd}, formalizes this statement.  

\begin{theorem} \label{thm:a-vwd-mc}
Given a $\delta \in (0, B]$, let $\mathcal{M}= (X, F)$ be a deterministic mechanism that implements $\alpha$-$\OMV_{MC}$ in very weakly 
dominant strategies for a $\delta$-uncertain-facility-location-game. Then, $\alpha \geq \frac{B}{2}$. 
\end{theorem}

Given the negative result, we move on to implementation in minimax dominant strategies in the hope of getting an analogous positive result as Theorem~\ref{thm:deltmedian}.

\subsection{Implementation in minimax dominant strategies}

When it comes to implementation in minimax dominant strategies, we again see that even in the case of the maxCost objective function one can do a lot better under this 
solution concept than under very weak dominance. But before we see the exact bounds one can obtain here, recall that Theorem~\ref{prop-onto} rules out the existence of 
mechanisms that are anonymous, minimax dominant, and onto. Hence, our focus will be on non-onto mechanisms. We note that the ideas in the following section can 
be broadly described as being similar to the ones in Section~\ref{sec:mdAC} since here, too, we focus on similar ``grid-based'' mechanisms. 

\subsubsection{Non-onto mechanisms}
In this section we show that there exists a mechanism, $\frac{\delta}{2}$-equispaced-phantom-half, that implements $\left(\frac{B}{4}+\frac{3\delta}{8}\right)$-$\OMV_{MC}$ 
in 
minimax dominant strategies. The mechanism is similar to the $\frac{\delta}{2}$-equispaced-median mechanism shown in Algorithm~\ref{algo:mech} and can be considered as an 
extension to the phantom-half mechanism proposed by \citet{golomb17}. Hence, we only highlight the changes in the description below. 

\textbf{$\frac{\delta}{2}$-equispaced-phantom-half.} Consider the mechanism described in Algorithm~\ref{algo:mech}. We need to make only two changes: i) instead of the 
definition of $A$ used in Algorithm~\ref{algo:mech}, we define it to be the set $\{g_1, \cdots, g_j, \cdots, g_m\}$, where $g_j = \frac{B}{2}, g_{i+1} - g_{i} = 
\frac{\delta}{2}$, for $1 \leq i \leq k-1$, $g_0 \geq 0$, and $g_m \leq B$. ii) instead of returning the median of the $l_i$s in line~\ref{step:med}, we return the 
median of the points $\ell_{min}, \frac{B}{2},$ and $\ell_{max}$, where $\ell_{min} = \min_i\{\ell_i\}$ and $\ell_{max} = \max_i\{\ell_i\}$.  

Below, we show that the $\frac{\delta}{2}$-equispaced-phantom-half mechanism described above implements $\left(\frac{B}{4}+\frac{3\delta}{8}\right)$-$\OMV_{MC}$ in 
minimax dominant strategies.

\begin{theorem} \label{thm:phantomHalf}
Given a $\delta \in [0, \frac{2B}{3}]$, the $\frac{\delta}{2}$-equispaced-phantom-half mechanism is anonymous and one that implements 
$\left(\frac{B}{4}+\frac{3\delta}{8}\right)$-$\OMV_{MC}$ in minimax dominant strategies for a $\delta$-uncertain-facility-location-game.
\end{theorem}

\begin{proof} [Proof (sketch)]
 The proof that the $\frac{\delta}{2}$-equispaced-phantom-half mechanism is minimax dominant is very similar to the proof of Lemma~\ref{lem:md} where we show that the 
$\frac{\delta}{2}$-equispaced-median mechanism is minimax dominant. Therefore, below we only show that it achieves an approximation bound of 
$(\frac{B}{4} + \frac{3\delta}{8})$.

To prove the bound, we begin by defining some notation and making a few observations. Let $\mathcal{L}_0 = \left([a_1, b_1], \cdots, [a_n, b_n] \right)$ be a profile of 
candidate locations of the agents and let $\{L_1, \cdots, L_n\}$ and $\{R_1, \cdots, R_n\}$ be the sorted order of the $a_i$s and $b_i$s, respectively. Also, let 
$\ell_i$ be the $\ell_i$ associated 
with the agent $i$ in the mechanism (see Algorithm~\ref{algo:mech}), $\ell_{min} = \min_i\{\ell_i\}$, and $\ell_{max} = \max_i\{\ell_i\}$. From the mechanism we know that 
$\ell_i \in [x_i, y_i]$, where $x_i$ and $y_i$ are as defined in the mechanism (see lines~\ref{step:defx}-\ref{step:defy} in Algorithm~\ref{algo:mech}). Additionally, for two 
agents $i$ and $j$, if $a_i \leq a_j$, then it is clear from the definition of $x_k$s (see line~\ref{step:defx} in Algorithm~\ref{algo:mech}) that $x_i \leq x_j$. Similarly, 
it again follows from the definition of $y_k$s that for two agent $i$ and $j$, if $b_i \leq b_j$, then $y_i \leq y_j$.

Next, consider the agent associated with $L_1$ (i.e., the agent who reports the smallest left endpoint). Without loss of generality we can assume that this agent is agent 1. 
From the discussion above we know that for every agent $i > 1$, we have that $\ell_i \geq x_i \geq x_1$ (since their left endpoint are at least $L_1$). Therefore, $\ell_{min} 
= \min_j\{\ell_j\} \geq x_1$, and since $x_1 \geq L_1 - 
\frac{\delta}{4}$ (this easily follows from the definition of $x_i$ and the fact that the points in $A$ are placed at a distance of $\frac{\delta}{2}$ apart), we have that 
$\ell_{min} \geq L_1 - \frac{\delta}{4}$. Now, one can employ a similar line of reasoning to see all of the following: $\ell_{max} \leq R_n + 
\frac{\delta}{4}$, $\ell_{min} \leq R_1 + \frac{\delta}{4}$, and $\ell_{max} \geq L_n - \frac{\delta}{4}$.

Given all the observations above, we are now ready to prove our bound. To do this, let us consider the following cases. In each of these cases, we will show that the 
output 
$p$ of the $\frac{\delta}{2}$-equispaced-phantom-half mechanism is at a distance of at most $(\frac{B}{4} + \frac{3\delta}{8})$ from $p_{opt}$, where from 
Proposition~\ref{prop:optSol-mc} we know that $p_{opt} = \frac{L_1 + R_1 + L_n + R_n}{4}$. 

\textbf{Case 1: $p = \frac{B}{2}$.} In this case, we have,
 \begin{align*}
 \abs{p_{opt} - p} &\leq \max\left(\frac{R_1 + R_{n}+ L_1 + L_{n}}{4} - p, p - \frac{R_1 + R_{n}+ L_1 + L_{n}}{4}\right)\\ 
  &= \max\left(\frac{R_1 - L_{1}}{4} + \frac{R_n + L_{n}}{4} + \frac{L_1}{2} - p, p - \frac{R_n}{2} + \frac{R_n - L_{n}}{4} - \frac{R_1 + L_{1}}{4}\right)\\
  &\leq \max\left(\frac{\delta}{4} + \frac{\ell_{min} + \frac{\delta}{4}}{2}, \frac{B}{2} - \frac{(\ell_{max} -\frac{\delta}{4})}{2} + \frac{\delta}{4}\right) 
\tag{{\scriptsize since  $p = \frac{B}{2}$,  $R_i - L_i \leq \delta$ (using Claim~\ref{clm:LRdelta}), $\ell_{min} \geq L_1 -\frac{\delta}{4}$, and $\ell_{max} \leq R_n + 
\frac{\delta}{4}$}}\\  
&\leq \max\left(\frac{3\delta}{8} + \frac{B}{4}, \frac{B}{4} + \frac{3\delta}{8}\right) \tag{{\scriptsize since $\ell_{min} \leq \frac{B}{2} \leq \ell_{max}$}}\\
&= \frac{B}{4} + \frac{3\delta}{8}
 \end{align*}

\textbf{Case 2: $p = \ell_{min}$.} In this case, we have,
 \begin{align*}
 \abs{p_{opt} - p} &\leq \max\left(\frac{R_1 + R_{n}+ L_1 + L_{n}}{4} - p, p - \frac{R_1 + R_{n}+ L_1 + L_{n}}{4}\right)\\ 
  &\leq \max\left(\frac{R_1 - L_{1}}{4} + \frac{R_n + L_{n}}{4} + \frac{L_1}{2} - p, p - \frac{R_1}{2} - \frac{L_1}{2}\right)
  \tag{{\scriptsize since $L_1 \leq L_n$ and $R_1 \leq R_n$}}\\
  &\leq \max\left(\frac{\delta}{4} + \frac{B}{2} + \frac{\delta}{8} - \frac{\ell_{min}}{2}, \frac{R_1}{2} - \frac{L_1}{2} + \frac{\delta}{4}\right) 
\tag{{\scriptsize since  $p = \ell_{min}$, $R_i - L_i \leq \delta$ (using Claim~\ref{clm:LRdelta}), $\ell_{min} \geq L_1 - \frac{\delta}{4}$, and $\ell_{min} \leq R_1 + 
\frac{\delta}{4}$}}\\  
&\leq \max\left(\frac{3\delta}{8} + \frac{B}{4}, \frac{3\delta}{4}\right) \tag{{\scriptsize since $\frac{B}{2} \leq \ell_{min} \leq \ell_{max}$}}\\
 &= \frac{B}{4} + \frac{3\delta}{8}  \tag{{\scriptsize since $\delta \leq \frac{2B}{3}$}}
 \end{align*}
 
\textbf{Case 3: $p = \ell_{max}$.} This can be handled analogously as in Case 2.

Finally, since from all the cases above we have that $\abs{p_{opt} - p} \leq (\frac{B}{4}+ \frac{3\delta}{8})$, one can easily see using 
Lemma~\ref{clm:maxregret-mc} that for $\mathcal{I} = [a_1, b_1] \times \cdots 
\times [a_n, b_n]$,  $\maxRegret(p, \mathcal{I}) - \text{OMV}_{MC}(\mathcal{I})$ is also bounded by $\left(\frac{B}{4}+\frac{3\delta}{8}\right)$.
\end{proof}

Given this result, it is natural to ask if we have a lower bound like the one in Section~\ref{sec:mdAC}. Unfortunately, the only answer we have is the obvious lower 
bound of $\frac{B}{4}$ that follows from the result of \citet[Theorem 15]{golomb17} who showed that under exact reports, and when using deterministic mechanisms, one cannot 
achieve a bound lower than $\frac{B}{4}$.

\section{Conclusions}

The standard assumption in mechanism design that the agents are precisely aware of their complete preferences may not be realistic in many situations. Hence, we 
believe that there is a need to look at models that account for partially informed agents and, at the same time, design mechanisms that provide robust guarantees. In this 
paper, we looked at such a model in the context of the classic single-facility location problem, where an agent 
specifies an interval instead of an exact location, and our focus was on designing robust mechanisms that perform well with respect to all the possible underlying true 
locations of the agents. Towards this end, we looked at two solution concepts, very weak dominance and minimax dominance, and 
we showed that, with respect to both the objective functions we considered, while it was not possible to achieve any good mechanism in the context of the former solution 
concept, extensions to the classical optimal mechanisms---i.e., mechanisms that perform optimally in the classical setting where the agents exactly know their 
locations---performed significantly better under the latter, weaker, solution concept.\footnote{In fact, if we 
consider an even weaker solution concept of undominated strategies (see \citep{chiesa12} for definitions), then it is not hard to show that the classical median and 
phantom-half mechanisms perform well under this solution concept (similar observations are also made, for instance, by \citet{chiesa15} with respect to the Vickrey mechanism 
in multi-unit auctions).} Our results are summarized in Table~\ref{tab1}.

There are some immediate open questions in the context of the problem we considered like looking at randomized mechanisms, providing tighter bounds, and potentially even 
finding deterministic mechanisms that perform better than the ones we showed. More broadly, we believe that it will be interesting to revisit the classic problems in 
mechanism 
design, see if one can look at models which take into account partially informed agents, and design mechanisms where one can explicitly relate the performance of the 
mechanism with the quality of preference information.

\printbibliography

\appendix

\section{Additional discussions} \label{app:addDiscussion}

\subsection{{Why regret?}} \label{app:whymR}

As stated in the introduction, our performance measure is based on minimizing the maximum regret. So, one question that could immediately arise is: 
``why regret?'' We argue below why this is a good measure by considering some alternatives.  

 \begin{enumerate}
  \item Perhaps one of the first approaches that comes to mind is to see if we can, for every possible input $I \in \mathcal{I}$, bound the ratio of the objective values of 
\emph{a)} the outcome that is returned by the mechanism and \emph{b)} the optimal outcome for that input. For instance, this is the approach taken by \citet{chiesa12} in the 
case of 
single 
good auctions. However, here this is not a good measure because we can quickly see that this ratio is always unbounded if there exists a point that is in the candidate 
locations of all the agents (i.e., if there is a $p \in [0, B]$ such that for all $i \in [n], p \in [a_i, b_i]$).

\item Another natural approach that could be taken is to show that for all possible inputs we can bound the difference between the objective values of \emph{a)} the outcome 
that is 
returned by the mechanism and \emph{b)} the optimal outcome for that input. For instance, this is the approach taken by \citet{chiesa15}. Technically, this is essentially 
what we 
are 
doing when using regret and finding an answer that has a max.\ regret that is additively close to the max.\ regret associated with the minimax optimal solution (one could 
argue in a similar way even when approximating multiplicatively---i.e., when finding an answer that has a max.\ regret that is multiplicatively close to the max.\ regret 
associated with the minimax optimal solution). The 
reason why using regret is more informative is because if we were to just mention that, for all $I \in \mathcal{I}$, the point $p$ that is returned 
by the mechanism satisfies,
\begin{equation*}
S(I,p) - S(I,p_{I}) \leq X, 
\end{equation*}
where $X$ is the bound we obtain, then the only information this conveys is that for every $I$ we are additively at most $X$-far from the optimal objective 
value, $p_{I}$, for $I$. However, instead, if we were to write it as 
\begin{equation*}
\maxRegret(p, \mathcal{I}) - \maxRegret (p_{opt}, \mathcal{I}) \leq Y, 
\end{equation*}
where $p_{opt}$ is the minimax optimal solution, then this conveys two things: a) for any point $p'$ there is at least one $I \in \mathcal{I}$ such that  
$S(I, p') - S(I, p_I) \geq Z$, where $Z =  \maxRegret (p_{opt}) = \text{optimal minimax value}$ (i.e., it gives us a sense on what is achievable at all---which in turn 
can be thought of as a natural lower bound) and b) the point $p$ that is chosen by the mechanism is at most $(Y+Z)$-far from the optimal objective value for any $I \in 
\mathcal{I}$. Hence, to convey these, we employ the notion of regret.  
 \end{enumerate}
 
 \subsection{Approximating additively vs multiplicatively} \label{app:addvsmul}
Even when working  with regret, when it comes to implementing a particular objective using some solution concept, one could potentially aim to find a solution $p$ such that 
$F = \frac{\maxRegret(p, \mathcal{I})}{\maxRegret(p_{opt},\mathcal{I})}$ is 
bounded (i.e., use a multiplicative approximation rather than additive). Although this is reasonable, there at at least two issues that become apparent: 
\begin{enumerate}
    \item When considering implementation in very weakly dominant strategies, it turns out that it is possible to show that there are no bounded mechanisms when using either 
of the objective functions (this is can be proved by proceeding like in the proof of Theorem~\ref{thm:a-vwd})
    
    \item When considering the objective of minimizing the maximum cost and minimax dominant strategies, it becomes very quickly clear that minimax dominant strategies are 
useless to look at as there are no bounded mechanisms. Why? Because suppose there was one. Then this implies that when the reports are exact---meaning every agent reports a 
single point---the mechanism should always return the optimal solution associated with this location profile, for if otherwise $F$ will not be bounded as minimum maximum 
regret when valuations are exact is zero. However, given that a minimax dominant mechanism is weakly dominant under exact reports, this in turn implies that we now have a 
mechanism that implements the optimal solution associated with the max.\ cost objective in weakly dominant strategies when the reports are 
exact. But then, we already know that there is no such mechanism due to a result by \citet[Theorem 3.2]{procaccia13}.
\end{enumerate}

Hence, we focus on additive approximations. 

\section{Revelation principle for minimax dominant strategies} \label{app:rev-pr}

Below we show that in the setting under consideration the revelation principle holds with respect to minimax dominant strategies.

\begin{lemma} \label{lem:rev-pr}
 Let $\mathcal{M}$ be a mechanism that implements a social choice function $f$ in minimax dominant strategies. Then, there exists a direct mechanism $\mathcal{M'}$ that 
implements $f$ and where for every agent $i$ reporting her candidate locations $K_i$ is a minimax dominant strategy.
\end{lemma}

\begin{proof}
 Let $(s_1, \cdots, s_n)$ be the minimax dominant strategy in $\mathcal{M}$ such that $f(K_1, 
\cdots, K_n) = F(s_1(K_1), \cdots s_n(K_n))$, where $F(\cdot)$ is the outcome function associated with $\mathcal{M}$. Next, let us define the outcome function, $F'$, 
associated with $\mathcal{M}'$ as 
 \begin{equation} \label{eqn:rev-pr2}
  F'(K_1, \cdots, K_n) = F(s_1(K_1), \cdots s_n(K_n)).
 \end{equation}
 
 Now, using the fact that $(s_1, \cdots, s_n)$ is a minimax dominant strategy in $\mathcal{M}$, we have that $\forall K_i', \forall K_{-i}$,
 \begin{multline} \label{eqn:rev-pr}
  \max_{\ell_i \in K_i} \max_{\sigma_{i} 
\in \Delta(\Sigma_i)} C(\ell_i, F(s_i(K_i), s_{-i}(K_{-i})) - C(\ell_i, F(\sigma_i(\ell_i), 
s_{-i}(K_{-i}))\\ 
\leq  \max_{\ell_i \in K_i} \max_{\sigma_{i} \in \Delta(\Sigma_i)} 
C(\ell_i, F(s_i(K_i'), s_{-i}(K_{-i})) - C(\ell_i, F(\sigma_i(\ell_i), s_{-i}(K_{-i}))).
 \end{multline}
 
 Additionally, if $\ell_i$ is the true location of agent $i$, then, again, using the fact that $(s_1, \cdots, s_n)$ is a minimax dominant strategy, we have that $\forall s_i' 
\in \Sigma_i, \forall K_{-i}$,
  \begin{equation}
  C(\ell_i, F(s_i(\ell_i), s_{-i}(K_{-i})) \leq C(\ell_i, F(s'_i(\ell_i), s_{-i}(K_{-i})).
 \end{equation}
 
 Therefore, 
 \begin{equation*}
  \min_{\sigma_{i} \in \Delta(\Sigma_i)}C(\ell_i, F(\sigma_i(\ell_i), s_{-i}(K_{-i})) = C(\ell_i, F(s_i(\ell_i), s_{-i}(K_{-i})),     
 \end{equation*}
 and using this in Equation~\ref{eqn:rev-pr} we have that $\forall K_i', \forall K_{-i}$,
 \begin{multline} \label{eqn:rev-pr1}
 \max_{\ell_i \in K_i} C(\ell_i, F(s_i(K_i), s_{-i}(K_{-i}))) - C(\ell_i, F(s_i(\ell_i), s_{-i}(K_{-i})))\\ 
\leq  \max_{\ell_i \in K_i} C(\ell_i, F(s_i(K_i'), s_{-i}(K_{-i}))) - C(\ell_i, F(s_i(\ell_i), s_{-i}(K_{-i}))).
 \end{multline}
 
 This in turn implies that using Equation~\ref{eqn:rev-pr2} we have that $\forall K_i', \forall K_{-i}$,
 \begin{equation*} \label{eqn:rev-pr3}
\max_{\ell_i \in K_i} C(\ell_i, F'(K_i, K_{-i})) - C(\ell_i, F'(\ell_i, K_{-i})) 
\leq  \max_{\ell_i \in K_i} C(\ell_i, F'(K_i', K_{-i})) - C(\ell_i, F'(\ell_i, K_{-i})),
 \end{equation*} 
 or in other words that reporting the candidate locations $K_i$ is a minimax dominant strategy in $\mathcal{M}'$. 
\end{proof}

\section{Additional claims} \label{app:addclms}
 
 \begin{claim} \label{clm:LRdelta}
 For all $i \in [n]$, $L_i \leq R_i$. {Additionally, $R_i - L_i \leq \delta$.}
\end{claim}

\begin{proof}
Note that the statement is true if we show that for every $L_i$ there are at most $(i-1)$ values in $\{R_k\}_{k\in [n]}$ such that they have a value less than $L_i$ (this is 
enough as this 
would imply that $R_i \geq L_i$). To see why that is true, consider the numbers $L_i, \cdots, L_n$ which are left endpoints of the reports of some agents. We know that each 
of these have a right endpoint associated with them (i.e., a $b'$, where the input is of the form $[L_i, b']$)  that are greater than them. Now, these constitute $(n - i + 
1)$ numbers and since there are only $n$ in total, there can at most $i-1$ of them that have value less than $L_i$.

To see the second part of the claim, consider the smallest $j$ such that $R_j - L_j > \delta$. Now, if $b$ is the right endpoint associated with $L_j$, then the fact that 
$R_j - L_j > \delta$ implies that $b < R_j$ (because otherwise $R_j - L_j \leq b - L_j \leq \delta$). Additionally, this in turn also implies that for some $L_k$, where $k 
\in 
\{1, \cdots, j-1\}$, there is a right endpoint $b'$ associated with $L_k$ such that $b' \geq R_j$. However, since $j$ is the smallest index value such that $R_j - L_j > 
\delta$, we now have a contradiction as $\delta < R_j - L_j < R_j - L_k \leq 
b' - L_k \leq \delta$.
\end{proof}

\begin{claim} \label{clm:lInterval}
 For any input $I = (\ell_1, \cdots, \ell_{n})$, $\ell_i \in [L_i, R_i]$.
\end{claim}

\begin{proof}
 Recall that, as mentioned under notations in Appendix~\ref{app:mos-ac}, the $\ell$s are in sorted order and $I$ is valid input if and only if there is a entry associated 
with every agent in 
it (i.e., for every agent 
$i$, $\exists j\colon \ell_j \in [a_i, b_i]$). Now, let us assume for the sake of contradiction that $\ell_i < L_i$. This implies that for $I$ to be a valid input, there has 
to be $(n-i+1)$ values in it that are 
greater than $\ell_i$---one from each agent $i$ who reported an interval $[a_i, b_i]$ such that $a_i\geq L_i$. However, this is not possible since there can only be at most 
$(n-i)$ values greater than $\ell_i$ as $\ell_i$ is the $i$th element in a sorted list.   

Similarly, one can argue analogously for the case when $\ell_i > R_i$. Therefore, $\ell_i \in [L_i, R_i]$.
\end{proof}
  
 \begin{claim} \label{clm:objeq}
 If $p_{opt}$ is a minimax optimal solution, then $obj_1^{AC}(p_{opt}) = obj_2^{AC}(p_{opt})$. 
\end{claim}

\begin{proof}
 For the sake of contradiction, let us assume without loss of generality that $obj_1^{AC}(p_{opt}) < obj_2^{AC}(p_{opt})$. Also, for the point $p_{opt}$, let $j$ be the 
smallest 
index such that $R_{j} > p_{opt}$ and $j\leq k$ (if no such $j$, then set $j = k+1$) and let $h$ be the largest index such that $L_{h} < p_{opt}$ and $h\geq k+2$ (if no 
such $h$, then set $h = k+1$). Note that if $n$ is even, then both $j$ and $h$ cannot be $k+1$, for if so then $obj_1^{AC}(p_{opt})$ would be equal to $obj_2^{AC}(p_{opt})$. 
% So below we assume that $h\geq k+2$. 
Now, consider the point $p = p_{opt} - \epsilon$, where 
$\epsilon = \frac{n(obj_2^{AC}(p_{opt}) - obj_1^{AC}(p_{opt}))}{\max((n - 2j + 2), (2h-n))}$, and let us compute $obj_1^{AC}(p)$ and $obj_2^{AC}(p)$. 
\begin{align*}
 obj_1^{AC}(p) &= \frac{1}{n}\left(2\sum_{i = j'}^k (R_i - p) + (n-2k)(R_{k+1} - p)\right) \tag{{\scriptsize $j'$ the smallest index such that $R_{j'} > p$ and $j'\leq k$ (if 
no such $j'$, then set $j' = k+1$)}} \\
 &= \frac{1}{n}\left(2\sum_{i = j}^k (R_i - (p_{opt} - \epsilon)) + (n-2k)(R_{k+1} - (p_{opt} - \epsilon))\right) \tag{{\scriptsize {note that $j' = j$}}\footnotemark} \\
 &= obj_1^{AC}(p_{opt}) + \frac{1}{n}\epsilon (n-2j+2) \\
 &\leq \frac{obj_1^{AC}(p_{opt}) + obj_2^{AC}(p_{opt})}{2} \\
 &< obj_2^{AC}(p_{opt}). 
\end{align*}

\footnotetext{because if not then $p_{opt}$ is not optimal since we can move to $R_{j'}\neq R_j$ and it can be verified that this point has regret less than that of 
$p_{opt}$.}

Computing $ obj_2^{AC}(p)$ similarly we see that $obj_2^{AC}(p) < obj_2^{AC}(p_{opt})$. 
% This in turn implies that, since $\epsilon > 0$ and $h \geq k+2$, we have that $obj_2^{AC}(p) < obj_2^{AC}(p_{opt})$. 
Now, since $obj_1^{AC}(p) < obj_2^{AC}(p_{opt})$ and $obj_2^{AC}(p) < obj_2^{AC}(p_{opt})$, this implies that $p_{opt}$ cannot be the minimax optimal solution 
as its maximum regret is larger than that of $p$'s.
\end{proof}

\begin{claim} \label{clm:atmost3inA}
In the $\frac{\delta}{2}$-equispaced-median mechanism, let $[a, b]$ be an interval of length at most $\delta$ and let $x$ and $y$ be the points that are 
closest (with ties being broken in favour of points in $[a, b]$ in both cases) to $a$ and $b$, respectively. Then, the interval $[x, y]$ can have at most 3 points 
that are also in $A$.    
\end{claim}

\begin{proof}
Suppose this were not the case and there existed two other points $x'< y'$ such that $x', y' \in A$ and $x', y' \in (x, y)$. Now, for this to happen $x$ has to be less than 
$a$ 
and $y$ has to be greater than $b$, for if otherwise then using the fact that $x$ and $y$ are the points that are the closest (with ties broken in favour of points in $[a, 
b]$) to $a$ and $b$, respectively, one can see that we can have only at most 3 points in $[x, y]$. However, this would imply that 
\begin{align*}
 2(b - a) &= (b- y' + y' - x' + x' - a) + (b-a)\\
 &> (y-b) + (y' - x') + (a-x) + (b-a) \tag{{\scriptsize $b - y' > y - b$ and $x' - a > a - x$, because $x, y$ are closest to $a, b$ and using the tie-breaking 
rule}}\\
 &= (y - x) + (y-x')\\
 &= 2\delta, \tag{{\scriptsize $y - x = y- y' + y' - x' + x' - x = \frac{3\delta}{2}$ as the distance between points in $A$ is $\frac{\delta}{2}$}}
\end{align*}
which in turn contradicts the assumption that $b - a \leq \delta$. 
\end{proof}

\begin{claim} \label{maxR@endp}
 Consider a mechanism $\mathcal{M} = (X, F)$ such that it is very weakly dominant for any agent to report her candidate locations if it is single point. Let $i$ be an agent 
with candidate locations $K_i = [a_i, b_i]$ and $p$ be the outcome of the mechanism for some profile of candidate locations $(K_i, K_{-i})$. Then, if $p_a$ and $p_b$ are the 
outcomes of $\mathcal{M}$ for the profiles $(a_i, K_{-i})$ and  $(b_i, K_{-i})$, respectively,
\begin{equation*}
 \maxRegret_i(p) = \max(|a_i-p| - |a_i-p_a|, |b_i-p| - |b_i-p_b|).
\end{equation*} 
% % Additionally, in such a mechanism, the maximum regret associated with a point inside an agent's interval is not more than the maximum regret associated with a point 
% % outside. 
\end{claim}

\begin{proof}
From Equation~\ref{eqn:maxRi} we know that 
\begin{equation*}
 \maxRegret_i(p) = \max_{\ell_i \in K_i} \max_{\sigma_{i} \in \Delta(\Sigma_i)} C(\ell_i, p) - C(\ell_i, F(\sigma_i(K_i), s_{-i}(K_{-i}))),
\end{equation*}
where $s_{-i}(K_{-i})$ is some set of actions played by the others.

Now, since it very weakly dominant for any agent to report her true candidate locations if it is a single point in $\mathcal{M}$, we have that 
\begin{equation*}
 \maxRegret_i(p) = \max_{\ell_i \in K_i} C(\ell_i, p) - C(\ell_i, p_{\ell})
\end{equation*}
where $p_{\ell}$ is the outcome of $\mathcal{M}$ for the profile $(\ell, K_{-i})$.

Given this, let us assume that the claim is false. This implies that there exists a $c \in (a_i, b_i)$ such that $C(c, p) - C(\ell_i, p_{c}) > C(a_i, p) - C(a_i, p_{a})$ and  
$C(c, p) - C(\ell_i, p_{c}) > C(b_i, p) - C(b_i, p_{b})$. However, we will show that this is impossible in both the cases below.

\textbf{Case 1: $c \leq p$.} In this case, we have, 
\begin{align*}
 C(c, p) - C(\ell_i, p_{c}) &= |c - p| - |c - p_c|\\
%  &= (p-c) - |c- p_c|\\
 &= (p - a_i) + (a_i - c) - |c- p_c|\\
 &= (p-a_i) - |a_i - p_a| + (a_i - c) +  |a_i - p_a| - |c- p_c|\\
 &= C(a_i, p) - C(a_i, p_{a}) + (a_i - c) +  |a_i - p_a| - |c- p_c|\\
 &\leq  C(a_i, p) - C(a_i, p_{a}) + (a_i - c) +  |a_i - p_c| - |c- p_c| \tag{{\scriptsize using very weak dominance for single reports}}\\
 &\leq  C(a_i, p) - C(a_i, p_{a}).
\end{align*}

\textbf{Case 2: $c > p$.} We can handle this analogously as in Case 1 and show that $C(c, p) - C(\ell_i, p_{c}) \leq  C(b_i, p) - C(b_i, p_{b})$. 
\end{proof}

\begin{claim} \label{clm:mr-ac-bounds}
 Let $p_{opt}$ be the minimax optimal solution associated with $\mathcal{I} = [a_1, b_1] \times\cdots\times [a_n, b_n]$ for the avgCost objective. Then, if $p$ is a point 
such that $\abs{p - p_{opt}} = d$ and $k = \lfloor \frac{n}{2} \rfloor$, then 
\begin{equation*}
 \frac{n-2k}{n}  \cdot  d \leq \maxRegret(p, \mathcal{I}) - \OMV_{AC}({\mathcal{I}}) \leq d.
\end{equation*}
 \end{claim}

\begin{proof}
If $\{L_i\}_{i \in [n]}$ and $\{R_i\}_{i \in [n]}$ denote the sorted order of the sets $\{a_i\}_{i \in [n]}$ and $\{b_i\}_{i \in [n]}$, respectively, then from 
Lemma~\ref{clm:maxregret} we know that for any point $p$ the maximum regret associated with $p$ can written as $\max(obj_1^{AC}(p),\allowbreak obj_2^{AC}(p))$, where
\begin{itemize}
 \item $obj_1^{AC}(p) = \frac{1}{n}\left(2\sum_{i = j}^k (R_i - p) + (n-2k)(R_{k+1} - p)\right)$, where $j$ is the smallest index such that $R_{j} > p$ and $j\leq k$ (if no 
such $j$, then set $j = k+1$)
 \item $obj_2^{AC}(p) = \frac{1}{n}\left(2\sum_{i = k+2}^h (p - L_i) + (n-2k) (p - L_{k+1})\right)$, where $h$ is the largest index such that $L_{h} < p$ and $h\geq k+2$ (if 
no such $h$, then set $h = k+1$).
\end{itemize}
Additionally, if $p_{opt}$ is the minimax optimal solution, we know from Claim~\ref{clm:objeq} that $\OMV_{AC}({\mathcal{I}}) = \allowbreak obj_1^{AC}(p_{opt})\allowbreak = 
obj_2^{AC}(p_{opt})$.

Now, let us assume without loss of generality that $p < p_{opt}$. Then, it clear that from above that $obj_1^{AC}(p) > obj_1^{AC}(p_{opt})$ and that $obj_2^{AC}(p) < 
obj_2^{AC}(p_{opt})$. So considering $obj_1^{AC}(p)$, we have 
\begin{align*}
 obj_1^{AC}(p) &= \frac{1}{n}\left(2\sum_{i = j}^k (R_i - p) + (n-2k)(R_{k+1} - p)\right)\\
   &= \frac{1}{n}\left(2\sum_{i = j}^{j'-1} (R_i - p) + 2\sum_{i = j'}^k (R_i - p)+  (n-2k)(R_{k+1} - p_{opt}) + (n-2k)(p_{opt} - p)\right) \tag{{\scriptsize where 
$j' \geq j$ is the smallest index such that $R_{j'} > p_{opt}$ }} \nonumber\\
&\leq \frac{1}{n}\left(2\sum_{i = j}^{j'-1} (p_{opt} - p) +  n \cdot obj_1^{AC}(p_{opt}) + 2\sum_{i = j'}^k (p_{opt} - p)  + (n-2k)(p_{opt} - p)\right) \tag{{\scriptsize 
since $R_i \leq p_{opt}$ for $i \leq j'-1$ and $n \cdot obj_1^{AC}(p_{opt}) = 2\sum_{i = j'}^k (R_i - p)+  (n-2k)(R_{k+1} - p_{opt})$ }} \nonumber\\
  &\leq \frac{1}{n}\left(2kd + (n-2k) d + n \cdot obj_1^{AC}(p_{opt})\right) \nonumber \tag{{\scriptsize $p_{opt} - p \leq d$ and $j \geq 1$}} 
\nonumber\\
  &= d + \OMV_{AC}({\mathcal{I}})  
\end{align*}

To see the upper bound, it is easy to see from above that one can set $j = k+1$ and hence we would have that $obj_1^{AC}(p) - \OMV_{AC}({\mathcal{I}}) \geq 
\frac{(n-2k)}{n}(p_{opt} - p) 
= \frac{(n-2k)}{n} \cdot d$.
\end{proof}

\section{\texorpdfstring{Minimax optimal solution for avgCost}{}} \label{app:mos-ac}
Given the candidate locations $K_i = [a_i, b_i]$ for all $i$, where, for some $\delta \in [0, B]$, $b_i - a_i \leq \delta$, in this section we are concerned with computing 
the minimax optimal solution $p_{opt}$ such that $p_{opt} = \argmin_{p \in [0, B]}\max_{I \in \mathcal{I}} \left(S(I, p) - \min_{p' \in {[0,B]}} S(I, p')\right)$, 
where $\mathcal{I} = [a_1, b_1] \times \cdots \times [a_n, b_n]$ and $S$ is the average cost function. Note that from the discussion in Section~\ref{sec:avgCost} we know that 
for any $I \in 
\mathcal{I}$, $\min_{p' \in {[0,B]}} S(I, p') = S(I, p_I)$, where $p_I$ is the median of the points in the vector $I$. Therefore, we can rewrite the definition of $p_{opt}$ 
as $p_{opt} = \argmin_{p \in [0, B]}\max_{I \in \mathcal{I}} S(I, p) - S(I, p_I)$. Next, before we move on further to computing $p_{opt}$, we introduce the following 
notation. 

\textbf{Notation:} Consider the left endpoints associated with all the agents, i.e., the set $\{a_i\}_{i \in [n]}$. We denote the sorted order of these points as 
$L_1, \cdots, L_{n}$ (throughout, by sorted order we mean sorted in non-decreasing order). Similarly, we denote the sorted order of the right endpoints, 
i.e., the points in the set $\{b_i\}_{i \in [n]}$, as $R_1, \cdots, R_{n}$. Additionally, for $i \in [n]$, we use $M_i$ to denote the mean of $L_i$ and 
$R_i$ (i.e., $M_i = \frac{L_i + R_i}{2}$). Throughout, for $\mathcal{I} = K_1\times\cdots\times K_n$, we refer to an element of $\mathcal{I}$ as an ``input'' and whenever we 
refer to an input $I \in \mathcal{I}$, where $I = (\ell_1, \cdots, \allowbreak \ell_{n})$, we assume without loss of generality that the $l_i$s are in sorted order (because 
the agents can always be re-indexed so that this is true). Also, given a point $p$, let $\mathcal{I}_1(p) \subseteq \mathcal{I}$ be the set of all inputs $I_1 = ( \ell_1, 
\cdots, \ell_{n})$ such that $\ell_{k+1} \geq p$. Similarly, let $\mathcal{I}_2(p) \subseteq \mathcal{I}$ be the set of all inputs $I_2 = ( \ell'_1, \cdots, \ell'_{n})$ such 
that $\ell'_{k+1} < p$. Often when the point $p$ is clear from the context, we write $\mathcal{I}_1$ and $\mathcal{I}_2$ to refer to $\mathcal{I}_1(p)$ and 
$\mathcal{I}_2(p)$, 
respectively.

Armed with the notation defined above, we can now prove the following lemma (which is also stated in Section~\ref{sec:avgCost}), which gives a concise formula to find the 
maximum regret associated with a point $p$ for the 
average cost 
objective. 

\maxRlemma*

\begin{proof}
 To prove this, consider the maximum regret associated with locating the facility at $p$ which is given by $\max_{I \in \mathcal{I}} S(I, p) - S(I, p_I)$. Given the fact 
that $\mathcal{I} = \mathcal{I}_1 \cup \mathcal{I}_2$, we can rewrite this as $\max(\max_{I_1 \in \mathcal{I}_1} S(I_1, p) - S(I_1, p_{I_1}), \max_{I_2 \in \mathcal{I}_2} 
S(I_2, p) - S(I_2, p_{I_2}))$. Now, let us consider each of these terms separately in the cases below.
 
 {\textbf{Case 1:} $\max_{I_1 \in \mathcal{I}_1} S(I_1, p) - S(I_1, p_{I_1})$.} Let us consider an arbitrary input $I_1=(\ell_1, \cdots, \allowbreak 
\ell_{n})$ that belongs to $\mathcal{I}_1$ (if $\mathcal{I}_1 = \emptyset$, then we define $\max_{I_1 \in \mathcal{I}_1} S(I_1, p) - S(I_1, p_{I_1}) = 0$). Now, the 
regret associated with $I_1$ is given by
 \begin{align*}
  \regret(p, I_1) &= S(I_1, p) - S(I_1, p_{I_1})\\
  &= \frac{1}{n}\left(\sum_{i=1}^{n} |\ell_i - p| - \left(\sum_{i=1}^{n} |\ell_i - \ell_{k+1}|\right)\right) \tag{{\scriptsize as $p_{I_1} = \ell_{k+1}$}}\\
  &= \frac{1}{n}\left(\sum_{i=1}^{k} |\ell_i - p| + \sum_{i=1}^{n-k} (\ell_{k+i} - p) - \left(\sum_{i=1}^{k} (\ell_{k+1} - \ell_i) + \sum_{i=2}^{n-k} (\ell_{k+i} - 
\ell_{k+1}) 
\right)\right) \tag{{\scriptsize as $\ell_{k+1} \geq p$ and $\ell_1\leq\cdots\leq\ell_{n}$}}\\
%   &=\sum_{i=1}^{k} |\ell_i - p| + \sum_{i=1}^{n-k} (\ell_{k+i} - p) - \left(-\sum_{i=1}^{k} \ell_i + \sum_{i=2}^{k+1} \ell_{k+i} \right) \\
  &= \frac{1}{n}\left(\sum_{i=1}^{k} \left(|\ell_i - p| + \ell_i - p\right) + (n-2k) (\ell_{k+1} - p)\right)\\
  &= \frac{1}{n}\left(\sum_{i=1}^{j-1} \left((p - \ell_i) + \ell_i - p\right) +  \sum_{i=j}^{k} \left((\ell_i - p) + \ell_i - p\right)+  (n-2k)(\ell_{k+1} - p) 
\right) \tag{{\scriptsize where $j$ is the smallest index such that $l_j > p$ and $j\leq k$}}\\
   &= \frac{1}{n}\left(2\sum_{i=j}^{k} (\ell_i - p) + (n-2k)(\ell_{k+1} - p)\right)
 \end{align*}  
 
 Therefore, if $obj_1^{AC}(p) = \max_{I_1 \in \mathcal{I}_1} S(I_1, p) - S(I_1, p_{I_1})$, then we have that
 \begin{equation*}
   obj_1^{AC}(p) = \max_{I_1 \in \mathcal{I}_1} S(I_1, p) - S(I_1, p_{I_1}) = \max_{I_1 \in \mathcal{I}_1} \allowbreak \frac{1}{n}\left(2\sum_{i=j}^{k} (\ell_i - p) 
+ \allowbreak (n-2k)(\ell_{k+1} - p)\right).
 \end{equation*}

And so, now since $\ell_i \in [L_i, R_i]$ (this is proved in Claim~\ref{clm:lInterval}, which is in Appendix~\ref{app:addclms}), we have that 
\begin{equation*}
obj_1^{AC}(p) = \frac{1}{n}\left(2\sum_{i = j}^k (R_i - p) + (n-2k)(R_{k+1} - p)\right), 
\end{equation*}
where $j$ is smallest index such that $R_{j} > p$ and $j\leq k$.
 
 \textbf{Case 2:} $\max_{I_2 \in \mathcal{I}_2} S(I_2, p) - S(I_2, p_{I_2})$. Like in the previous case, consider an arbitrary input $I_2 = ( \ell'_1, \cdots, 
\allowbreak \ell'_{n})$ that belongs to $\mathcal{I}_2$ (if $\mathcal{I}_2 = \emptyset$, then we define $\max_{I_2 \in \mathcal{I}_2} S(I_2, p) - S(I_2, p_{I_2}) = 
0$). Now, doing a similar analysis as in Case 1, we will see that the regret associated with $I_2$ is given by $\frac{1}{n}\left(2\sum_{i = k+2}^h (p - \ell'_i) + (n-2k)(p - 
\ell'_{k+1})\right)$, 
where $h$ 
is the largest index such that $\ell_{h} < p$ and $h\geq k+2$. Therefore, $obj_2^{AC}(p) = \max_{I_2 \in \mathcal{I}_2} S(I_2, p) - S(I_2, p_{I_2}) = 
\frac{1}{n}\left(2\sum_{i 
= k+2}^h (p - 
L_i) + (n-2k)(p - L_{k+1})\right)$, where $h$ is largest index such that $L_{h} < p$ and $h\geq k+2$. 
 
 Hence, combining both the cases we have that the maximum regret associated with $p$ is $\max(obj_1^{AC}(p), obj_2^{AC}(p))$.
\end{proof}

Using the lemma proved above, one can show that for the minimax optimal solution, $p_{opt}$, $obj_1^{AC}(p_{opt}) \allowbreak = obj_2^{AC}(p_{opt})$ (this is proved in 
Claim~\ref{clm:objeq}, which is in Appendix~\ref{app:addclms}). And this observation in turn brings us to our next lemma (which is, again, also stated in 
Section~\ref{sec:avgCost}) which shows that 
$p_{opt}$ is always in the interval $[L_{k+1}, R_{k+1}]$. 

\optInterval*

\begin{proof}
 Let us assume for the sake of contradiction that $p_{opt} < L_{k+1}$ or $p_{opt} > R_{k+1}$. We will consider each of these cases separately and show that for both the cases 
$M_{k+1}$ has a lesser maximum regret, thus contradicting the fact that $p_{opt}$ is the minimax optimal solution.
 
 \textbf{Case 1: $p_{opt} < L_{k+1}$.} From Lemma~\ref{clm:maxregret} we know that $n \cdot obj_1^{AC}(M_{k+1}) = 2\sum_{i = j}^k (R_i - M_{k+1}) + (n-2k) (R_{k+1} - 
M_{k+1})$, where 
$j$ is 
smallest index such that $R_{j} > M_{k+1}$ and $j\leq k$ (if there is no such $j$, then set $j = k +1$). Now, consider the input $I_1 = (R_1, \cdots, R_{2k+1})$ that belongs 
to $\mathcal{I}_1(p_{opt})$ and let us 
calculate the regret associated with $p_{opt}$ for $I_1$. 
 \begin{align*}
  \regret(p_{opt}, I_1) &= \frac{1}{n}\left(2\sum_{i = j'}^k (R_i - p_{opt}) + (n-2k) (R_{k+1} - p_{opt})\right) \tag{{\scriptsize $j'$ is the smallest index such that $R_j' 
> 
p_{opt}$ and $j' \leq k$}}\\
  &\geq \frac{1}{n}\left(2\sum_{i = j}^k (R_i - p_{opt}) + (n-2k) (R_{k+1} - p_{opt})\right) \tag{{\scriptsize $j' \leq j$ since $R_j > M_{k+1} > p_{opt}$}}\\
  &> \frac{1}{n}\left(2\sum_{i = j}^k (R_i - L_{k+1}) + (n-2k) (R_{k+1} - p_{opt})\right) \tag{{\scriptsize as $p_{opt} < L_{k+1}$}}\\
  &> \frac{1}{n}\left(2\sum_{i = j}^k (R_i - (2M_{k+1} - R_{k+1})) + (n-2k)(R_{k+1} - (2M_{k+1} - R_{k+1}))\right) \tag{{\scriptsize as $L_{k+1} = 2M_{k+1} - R_{k+1}$}}\\
  &=obj_1^{AC}(M_{k+1}) + \frac{1}{n}(n-2j+2)(R_{k+1} - M_{k+1}).
 \end{align*}
 Now, since $p_{opt}$ is the optimal solution and $\maxRegret(p_{opt}, \mathcal{I}) = obj_1^{AC}(p_{opt})$ (from Claim~\ref{clm:objeq}), we have that $obj_1^{AC}(p_{opt}) 
\geq 
\regret(p_{opt}, I_1)$, which in turn implies that from above we have 
 \begin{equation} \label{eq:obj1}
  obj_1^{AC}(p_{opt}) > obj_1^{AC}(M_{k+1}) + \frac{1}{n}(n-2j+2)(R_{k+1} - M_{k+1}).
 \end{equation}

 Next, let us consider $obj_2^{AC}(p_{opt})$. Note that, given Claim~\ref{clm:lInterval} and since $p_{opt} < L_{k+1}$, $\mathcal{I}_2(p_{opt}) = \emptyset$. Therefore, 
$obj_2^{AC}(p_{opt}) = 0$, and so from Claim~\ref{clm:objeq} we know that since $p_{opt}$ is an optimal solution, $obj_1^{AC}(p_{opt}) = obj_2^{AC}(p_{opt}) = 0$. However, 
$j\leq k+1$, and 
hence from Equation~\ref{eq:obj1} we have that $obj_1^{AC}(p_{opt}) > 0$, thus in turn contradicting the fact that $p_{opt}$ is the minimax optimal solution.  

  \textbf{Case 2: $p_{opt} > R_{k+1}$.} We can do a similar analysis as in Case 1 to again show that this cannot be the case. 
  
 Hence, from both the cases above, we have that $p_{opt} \in [L_{k+1}, R_{k+1}]$.
\end{proof}

\begin{algorithm}[tb]
\noindent\fbox{%
\begin{varwidth}{\dimexpr\linewidth-4\fboxsep-4\fboxrule\relax}
\begin{algorithmic}[1]
%   {
 \footnotesize
  \Input For each agent $i$, their input interval $[a_i, b_i]$ 
  \Output Minimax optimal solution for the given instance
  \State $p_{opt} \leftarrow 0$, $mR \leftarrow \infty$
  \State Let $\{L_i\}_{i\in[n]}$ be the sorted order the points in $\{a_i\}_{i\in[n]}$ and $\{R_i\}_{i\in[n]}$ be the sorted order the points in $\{b_i\}_{i\in[n]}$
  \State Let $C = \{R_i, L_j | i\leq k+1, j\geq k+1, L_{k+1} < R_i, L_j < R_{k+1}\}$ and $H = \{h_1, \cdots, h_{|C|}\}$ be the sorted order of the points in $C$.  
  \For{each $h_i \in H$, where $i\in \{1, \cdots, |C|\}$}
    \State $x_i \leftarrow \#(R_j)$, where $R_j \geq h_i$ and $j\leq k$
    \State $y_i \leftarrow \#(L_j)$, where $L_j \leq h_i$ and $j\geq k+2$
    \State $S^1_i \leftarrow \sum (R_j)$, where $R_j \geq h_i$ and $j\leq k$
    \State $S^2_i \leftarrow \sum (L_j)$,where $L_j \leq h_i$ and $j\geq k+2$  
  \EndFor 
  \For{each $[h_i, h_{i+1}]$, where $i\in \{1, \cdots, |C|\}$}
    \If{$obj_1^{AC}(h_i) == obj_2^{AC}(h_i)$ and $ obj_1^{AC}(h_i) < mR$} \Comment{{\footnotesize Is $h_i$ a possible solution?}}
      \State $p_{opt} \leftarrow h_i, mR \leftarrow obj_1^{AC}(h_i)$
   \EndIf
    \If{$obj_1^{AC}(h_{i+1}) == obj_2^{AC}(h_{i+1})$ and $obj_1^{AC}(h_{i+1}) < mR$} \Comment{{\footnotesize Is $h_{i+1}$ a possible solution?}}
      \State $p_{opt} \leftarrow h_{i+1}, mR \leftarrow obj_1^{AC}(h_{i+1})$
   \EndIf 
%    \Statex 
   \State $p_i \leftarrow \frac{(n-2k) M_{k+1} + S^1_{i+1} + S^2_i}{x_{i+1} + y_i + (n-2k)}$ \Comment{{\footnotesize If $p_{opt} \in (h_i, h_{i+1})$, then use the fact that
$obj_1^{AC}(p_{opt}) = obj_2^{AC}(p_{opt})$}}
   \If{$p_i \in (h_i, h_{i+1})$ and $obj_1^{AC}(p_i) < mR$} \Comment{{\footnotesize Is $p_{i}$ feasible and is it a possible solution?}}
      \State $p_{opt} \leftarrow p_i, mR \leftarrow obj_1^{AC}(p_i)$
   \EndIf   
  \EndFor
  \State return $p_{opt}$ 
%   }
\end{algorithmic}
\end{varwidth}
}
\caption{Computing the minimax optimal solution.}
\label{algo1}
\end{algorithm}
%  }

Equipped with the lemmas proved above, we are now ready to compute the minimax optimal solution.

\begin{proposition} \label{prop:optSol}
 Algorithm~\ref{algo1} computes $p_{opt}$ in $O(n\log n)$ time.
\end{proposition}

\begin{proof}[Proof (sketch)]
 We know from Lemma~\ref{clm:optInterval} that $p_{opt} \in [L_{k+1}, R_{k+1}]$. So, all we are doing in Algorithm~\ref{algo1} is to consider all the points in $C = \{R_i, 
L_j | i\leq k+1, j\geq k+1, R_i > L_{k+1}, L_j < R_{k+1}\}$ in sorted order (which is the set $H$ in the algorithm) and check if for every interval $[h_i, h_{i+1}]$, whether 
$p_{opt} = h_i$, $p_{opt} = h_{i+1}$, or $p_{opt} \in (h_i, h_{i+1})$ (lines 10--21). In the last case, since there are no points in $C$ that are between $h_i$ and $h_{i+1}$, 
we can use the fact that for an optimal point $p_{opt}$, $obj_1^{AC}(p_{opt}) = obj_2^{AC}(p_{opt})$ (from Claim~\ref{clm:objeq}, which is in Appendix~\ref{app:addclms}) to 
obtain a value of $p_i$ (line 17). In 
line 18 we 
just check 
if this point actually lies in $(h_i, h_{i+1})$ and also see if it is better than the best solution we currently have. 

One can see that through the cases we consider in Algorithm~\ref{algo1} we are trying out all possible values the minimax optimal solution can take and hence the algorithm is 
correct. Also, it is easy to see that this can be done in $O(n\log n)$ time.
\end{proof}

\section{Minimax optimal solution for maxCost} \label{app:mos-mc}
Given the candidate locations $[a_i, b_i]$ for all $i$,  where, for some $\delta \in [0, B]$, $b_i - a_i \leq \delta$, here we are concerned with computing the minimax 
optimal solution $p_{opt}$ such that $p_{opt} = \argmin_{p \in [0, B]} \allowbreak \max_{I \in \mathcal{I}} \left(S(I, p) - \min_{p' \in {[0,B]}} S(I, p')\right)$, where 
$\mathcal{I} = [a_1, b_1] \times \cdots \times [a_n, b_n]$ and $S$ is the maximum 
cost function. Note that from the discussion in Section~\ref{sec:maxCost} we know that if $I = (\ell_1, \cdots, \ell_n)$ (as stated in Section~\ref{sec:avgCost}, we can 
assume without loss of 
generality that the $l_i$s are in sorted order) is a valid input in $\mathcal{I}$, then $\min_{p' \in {[0,B]}} S(I, p') = S(I, p_I)$, where $p_I = \frac{\ell_1 + \ell_n}{2}$.
Therefore, we can rewrite the definition of $p_{opt}$ as $p_{opt} = \argmin_{p \in [0, B]}\max_{I \in \mathcal{I}} S(I, p) - S(I, 
p_I)$. 

Next, we prove the following lemma.

\begin{restatable}{lemma}{maxregretMC} \label{clm:maxregret-mc}
 Given a point $p$, the maximum regret associated with $p$ for the maximum cost objective can written as $\max(obj^{MC}_1(p),\allowbreak obj^{MC}_2(p))$, where
\begin{itemize}
 \item $obj^{MC}_1(p) = \frac{R_1 + R_{n}}{2} - p$
 \item $obj^{MC}_2(p) = p - \frac{L_1 + L_{n}}{2}$.
\end{itemize}
\end{restatable}

\begin{proof}
 Consider the maximum regret associated with locating the facility at $p$, which is given by $\max_{I \in \mathcal{I}} S(I, p) - S(I, p_I)$, and the sets $\mathcal{I}'_1, 
\mathcal{I}'_2$, where $\mathcal{I}'_1 \subseteq \mathcal{I}$ is the set of all inputs such that for every $I'_1 = (\ell_1, \cdots, 
\ell_{n})$ that belongs to $\mathcal{I'}_1$, $\frac{\ell_1 + \ell_{n}}{2} \geq p$ and $\mathcal{I}'_2 \subseteq \mathcal{I}$ is the set of all inputs such that for 
every $I'_2 = \langle \ell'_1, \cdots, 
\ell'_{n}\rangle$ that belongs to $\mathcal{I}'_2$, $\frac{\ell'_1 + \ell'_{n}}{2} < p$. Since $\mathcal{I} = \mathcal{I}'_1 \cup \mathcal{I}'_2$, we can rewrite this as 
$\max(\max_{I'_1 \in \mathcal{I}'_1} S(I'_1, p) - S(I'_1, p_{I'_1}), \max_{I'_2 \in \mathcal{I'}_2} 
S(I'_2, p) - S(I'_2, p_{I'_2}))$. Now, let us consider each of these terms separately in the cases below.
 
 {\textbf{Case 1:} $\max_{I'_1 \in \mathcal{I}'_1} S(I'_1, p) - S(I'_1, p_{I'_1})$.} Let us consider an arbitrary input $I'_1 = \langle \ell_1, \cdots, \allowbreak 
\ell_{n}\rangle$ that belongs to $\mathcal{I}'_1$ (if $\mathcal{I}'_1 = \emptyset$, then we define $\max_{I'_1 \in \mathcal{I}'_1} S(I'_1, p) - S(I'_1, p_{I_1}) = 0$). Now, 
the regret associated with $I'_1$ is given by
\begin{align*} 
  \regret(p, I'_1) &= S(I'_1, p) - S(I'_1, p_{I'_1})\\
  &= \max(|\ell_1 - p|, |\ell_n - p|) - \left|\ell_n - \frac{\ell_1 + \ell_{n}}{2}\right|\nonumber\\
  &= (\ell_n - p) - \left(\ell_n - \frac{\ell_1 + \ell_{n}}{2}\right) \tag{{\scriptsize as $p \leq \frac{\ell_1 + \ell_{n}}{2}$}} \nonumber\\
  &= \frac{\ell_1 + \ell_{n}}{2} - p.
 \end{align*}
 
 Therefore, making use of the fact that $\ell_i \in [L_i, R_i]$ (this is proved in Claim~\ref{clm:lInterval}, which is in the appendix), we have that 
 \begin{align}\label{eqn:mc1}
  \max_{I'_1 \in \mathcal{I}'_1} S(I'_1, p) - S(I'_1, p_{I'_1}) &= \max_{I'_1 \in \mathcal{I}'_1}\left(\frac{\ell_1 + \ell_{n}}{2} - p\right) = \frac{R_1 + R_{n}}{2} - p = 
obj^{MC}_1(p). 
 \end{align}
 
  \textbf{Case 2:} $\max_{I'_2 \in \mathcal{I}'_2} S(I'_2, p) - S(I'_2, p_{I'_2})$. Like in the previous case, consider an arbitrary input $I'_2 = \langle \ell'_1, 
\cdots,\allowbreak \ell'_{n}\rangle$ that belongs to $\mathcal{I}'_2$ (if $\mathcal{I}'_2 = \emptyset$, then we define $\max_{I'_2 \in \mathcal{I}'_2} S(I'_2, p) - S(I'_2, 
p_{I'_2}) = 0$). Now, doing a similar analysis as in Case 1, we will see that the regret associated with $I'_2$ is given by $p - \frac{\ell'_1 + \ell'_{n}}{2}$. Hence, 
again, making use of the fact that $\ell_i \in [L_i, R_i]$ (see Claim~\ref{clm:lInterval}), we have that
  \begin{align}\label{eqn:mc2}
  \max_{I'_2 \in \mathcal{I}'_2} S(I'_2, p) - S(I'_2, p_{I'_2}) &= \max_{I'_2 \in \mathcal{I}'_2}\left(p - \frac{\ell'_1 + \ell'_{n}}{2}\right) = p - \frac{L_1 + L_{n}}{2} = 
obj^{MC}_2(p). 
 \end{align}
 
Now, from equations~\ref{eqn:mc1} and~\ref{eqn:mc2} we have our lemma.
\end{proof}

Equipped with the lemma proved above, we can now find the minimax optimal solution for the maximum cost objective.

\begin{restatable}{proposition}{optMC} \label{prop:optSol-mc}
 The minimax optimal solution, $p_{opt}$, for the maximum cost objective is given by $p_{opt} = \frac{L_1 + R_1 + L_n + R_n}{4}$.
\end{restatable}

\begin{proof}[Proof (sketch)]
 Note that the statement immediately follows if we prove that for a minimax optimal solution, $p_{opt}$, $obj^{MC}_1(p_{opt}) = obj^{MC}_2(p_{opt})$. And the latter 
proposition 
is easy to see, for if $obj^{MC}_1(p_{opt}) \neq obj^{MC}_2(p_{opt})$, then the maximum regret associated with  $p_{opt}$ is greater than $\frac{R_1+R_n-L_1-L_n}{4}$, 
whereas the maximum regret associated with the point $\frac{L_1 + R_1 + L_n + R_n}{4}$ is exactly $\frac{R_1 + R_n - L_1 - L_n}{4}$.
\end{proof}

% \begin{claim} \label{clm:mr-mc-bounds}
%  Let $p_{opt}$ be the optimal minimax solution associated with $\mathcal{I} = [a_1, b_1] \times\cdots\times [a_n, b_n]$ for the maxCost objective. Then, if $p$ is a point 
% such that $\abs{p - p_{opt}} 
% = d$, then $\frac{d}{n} \leq \maxRegret(p, \mathcal{I}) - \OMV_{MC}{\mathcal{I}}) \leq d$.  
% \end{claim}
% 
% \begin{proof}
%  
% \end{proof}

\end{document}